\documentclass[10pt, journal, letterpaper]{IEEEtran}
\ifCLASSOPTIONcompsoc
   \usepackage[tight,normalsize,sf,SF]{subfigure}
\else
   \usepackage[tight,normalsize]{subfigure}
\fi

\usepackage{booktabs}

\usepackage{color,ifthen,graphicx,enumerate,float,url}
\usepackage{mathrsfs}
\usepackage{epstopdf}
\usepackage{nicefrac}
\usepackage{pstool}
\usepackage{subfigure}
\usepackage{caption}

\usepackage{amsmath,amssymb}
\usepackage{bm}
\usepackage{breakurl}
\usepackage{caption}
\usepackage{latexsym}
\usepackage{amsmath}
\usepackage{amsfonts}
\usepackage{amsthm}
\usepackage{extarrows}
\usepackage{multirow}
\usepackage{epsfig}
\usepackage{indentfirst}
\usepackage{multirow}
\usepackage{algorithm}
\usepackage{algorithmic}
\usepackage{float}

\newtheorem{theorem}{Theorem\,}
\newtheorem{fact}{Fact\,}
\newtheorem{definition}{Definition\,}
\newtheorem{proposition}[theorem]{Proposition\,}

\newtheorem{assumption}[theorem]{Assumption\,}

\newtheorem{conjecture}[theorem]{Conjecture\,}
\newtheorem{coro}[theorem]{Corollary\,}

\newcommand{\dif}{\mathrm{d}}

\begin{document}

\title{Recommending Paths: Follow or Not Follow?}
\author{
Yunpeng~Li,~Costas~Courcoubetis,~and~Lingjie~Duan
\thanks{Y. Li, C. Courcoubetis and L. Duan are with the Engineering Systems and Design Pillar, Singapore University of Technology and Design, Singapore 487372 (e-mail: yunpeng\_li@mymail.sutd.edu.sg, \{costas, lingjie\_duan\}@sutd.edu.sg).}
}
\maketitle

\begin{abstract}
Mobile social network applications constitute an important platform for traffic information sharing,
helping users collect and share
sensor information about the driving conditions they experience on the traveled path in real time.
In this paper we analyse the simple but fundamental model
of a platform choosing between two paths: one
with known deterministic travel cost and the other
that alternates over time between a low and a high random cost states,
where the low and the high cost states are only partially observable and perform respectively better and worse
on average than the fixed cost path.
The more users are routed over the stochastic path,
the better the platform
can infer its actual state and use it efficiently.

At the Nash equilibrium, if asked to take the riskier path, in many cases selfish users (that are allowed to have access to the information collected by the platform) will myopically disregard the optimal path suggestions of the platform, leading to a suboptimal system without enough exploration on the stochastic path.
We prove the interesting result that if the past collected information is hidden from users,
the system becomes incentive compatible and even `sophisticated' users
(in the sense that they have full capability to reverse-engineer the platform's recommendation
and derive the path state distribution conditional on the recommendation)
prefer to follow the platform's recommendations.
In a more practical setting where the platform implements a model-free Q-learning algorithm to
minimise the social travel cost,
our analysis suggests that increasing the accuracy of the learning algorithm increases
the range of system parameters for which sophisticated users follow the recommendations of the platform,
becoming in the limit fully incentive compatible.
Finally, we extend the two-path model to include more stochastic paths, and show that incentive compatibility
holds under our information restriction mechanism.
\end{abstract}

\section{Introduction}

Given millions of inter-connected smartphones\footnote{
Smartphones are equipped with various sensors such as camera, GPS and accelerometer which enable mobile users to easily sense
many real-time traffic conditions when they drive \cite{bhoraskar2012wolverine}.}
and in-vehicle sensors sold annually,
it is promising to leverage the crowd for data sensing and sharing.
 Mobile social network applications constitute an important platform for traffic information sharing,
helping users collect and share
real-time sensor information about the driving conditions they experience on the traveled path,
see \cite{lequerica2010drive}.
Platforms inform new travelers of the paths they should take
by aggregating information from other users that used these paths in the past and
recommending a path with the least estimated current cost for travelling.
 For example, Waze
 uses a mobile social network
 platform for drivers to share traffic and road information.
 Another example is Google Map which uses real-time traffic data shared by
 hundreds of millions of people around the world
 to analyse traffic and road conditions \cite{google}.

All these platforms estimate the \emph{current} average cost of the
alternative paths and suggest the least costly paths to their travelers.
Obviously, a selfish user will follow such a myopic suggestion. But would she
follow the suggestions of an optimal (social cost optimising in the long run) platform that
frequently explores riskier paths in case these become superior over time?
This incentive issue becomes even more important since
our Price of Anarchy analysis (see Section \ref{myopic}) suggests
that myopic platforms whose recommendations users are likely to follow
can be arbitrarily bad in term of efficiency compared to the optimal platform.
%
%

In this paper we illustrate the above issues by considering the simple but
fundamental case of jointly routing and learning in a context where users
decide their trips from point $A$ to point $B$ by choosing between two paths P1 and P2.
P2 has a fixed user driving cost whereas P1 has driving conditions (e.g., visibility, `black ice' segments, congestion)
that alternate between a `good' and a `bad' states according to a two-state partially observable Markov chain with known transition probabilities,
influencing the expected driving cost over the path.
When in good (bad) condition, P1 has lower (higher) expected cost than P2.
By aggregating information about the actual cost experienced by users that traveled over P1, a mobile platform can
estimate its current state, and make the appropriate recommendation to future travelers.
Selfish users deciding on their current trip would prefer P1 only if its current expected cost
conditioned on the available information is less than the known cost of P2.

But there are additional reasons to explore P1 even if it momentarily looks on average worse than P2.
An `altruistic' user would take this ex-ante costlier path in order to increase system information about P1.
With little luck, finding P1 in its good state will benefit future travelers which will exploit this information.
Hence a socially optimal platform would advise at appropriate times some of the users to use paths
that are myopically suboptimal to them.
Unfortunately, this is not a Nash equilibrium strategy for the system since
without appropriate incentives, selfish users will always choose
the path with the least current expected cost.
This results in exploring the stochastic path P1 less frequently than socially desired.

We show that the myopic routing strategy achieves a Price of Anarchy (PoA)
that can be arbitrarily large
compared to the
case that users follow the recommendations of the optimal platform.
We prove that by restricting information the incentives of the users become aligned with the
incentives of the social planner: simply hide the information reported by the past travellers
and recommend the socially optimal path choice to current travelers.
Using a correlated equilibrium concept, we show that the
equilibrium strategy of the users is to follow the recommendations of the optimal platform.

Numerous works have been done on traffic estimation based on information sharing by travelers (e.g., \cite{li2009performance}, \cite{castro2012urban}).
Our paper does not deal with technical details on how to aggregate and process information or
on how to architect such systems. It provides a simple conceptual model for
user incentive mechanism design when there are exploration-exploitation trade-offs.
On a different direction, exploration-exploitation in optimal decision making  is well studied in classical multi-armed bandit problems
where decisions are made centrally (e.g., \cite{gittins2011multi}).
In our  model we have multiple bandits (corresponding to paths) but we cannot force the optimal
sequence of choosing arms due to users' selfishness. Each machine will be played in a myopic sense if full information is disclosed.

Incentive mechanism design for participation in crowdsensing platforms
has been well studied recently (e.g., \cite{yang2012crowdsourcing}, \cite{chorppath2013trading}, \cite{duan2012incentive}).
In our case participation is not an issue since we prove that users always gain by participating.
As a parallel to the case of allocating tasks to agents, our goal is to
incentivise agents to accept tasks that may not be optimal for them, but create the best results
for the rest of the community.
Similar to our idea, in the economics literature there are some recent work
(\cite{kremer2014implementing,frazier2014incentivizing}) for motivating the wisdom of the crowd.
Yet, \cite{kremer2014implementing} did not look at a dynamic Markov chain model for long-term forecasting
and \cite{frazier2014incentivizing} requires incentive payments (which is not possible for many traffic recommendation applications).
Instead, we model and analyze a more interesting but complex partially observable Markov decision process (POMDP),
and propose a payment-free incentive mechanism for the POMDP model.
Further, we study the incentive compatibility of model-free reinforcement learning,
which approximates the complex POMDP policy and is easy to implement in practice.
Our main contributions are:

\begin{itemize}
\item We formulate a joint routing and learning model for users making travel path choices.
The POMDP model is simple but powerful enough to formulate some key problems
in incentive compatible platform design.
The optimal policy for recommending paths may prefer paths with higher average costs
to exploit their low cost states. This policy serves as a benchmark for
efficiency comparisons with other policies.
\item
Although the optimal policy cannot be derived in closed form, we compute the
Price of Anarchy (PoA) of myopic decision making by comparing to the optimum.
If platforms (or users) minimise the
short term travel cost,
PoA is equal to $1/(1-\beta)$, where $\beta\in(0,1)$ is the discount factor used in the optimal policy. This tells that myopic platforms whose recommendations
users are likely to follow can be arbitrarily bad.
\item
We consider the challenging case of `sophisticated' users: such a user
has full system information (i.e., system parameters and the used
POMDP to derive the optimal policy).
If we allow such users to access the travel information collected by the platform from past travelers,
the system with sophisticated users has an equilibrium that corresponds to using the myopic policy. Accordingly, we propose an information restriction mechanism such that the equilibrium is to follow the recommendations of
the optimal policy, achieving PoA =1.
\item
In practice, an approximation of the optimal policy can be obtained via  reinforcement learning.
We consider the  incentive compatibility of the platform using Q-learning.
We numerically  show that the more accurate the learning algorithm is,
the `more' incentive compatible the system with restricted information becomes.
We further extend the two-path model to include more stochastic paths, and show that the incentive compatibility is easier to ensure under our information restriction mechanism.\end{itemize}

The rest of the paper is organized as follows. Section~\ref{systemmodel} introduces the network model
and formulates the problem as a POMDP over a belief state about the paths.
Section \ref{optimal} presents the optimal platform design
and Section \ref{myopic} presents two myopic platforms as comparison benchmarks.
Section \ref{mechanism} shows the incentive mechanism design for myopic users.
Section \ref{RL} presents the model-free optimization technique of Q-learning
and analyses the incentive compatibility issues, and Section \ref{multipath} extends the two-path model for examining users' incentive compatibility.
Section \ref{conclusion} concludes.


\section{System Model and Problem Formulation}\label{systemmodel}

\begin{figure*}
\centering
\subfigure[A two travel path network]{
\begin{minipage}[b]{0.45\linewidth}\label{tp}

  \centering
  \includegraphics[width=1.9in]{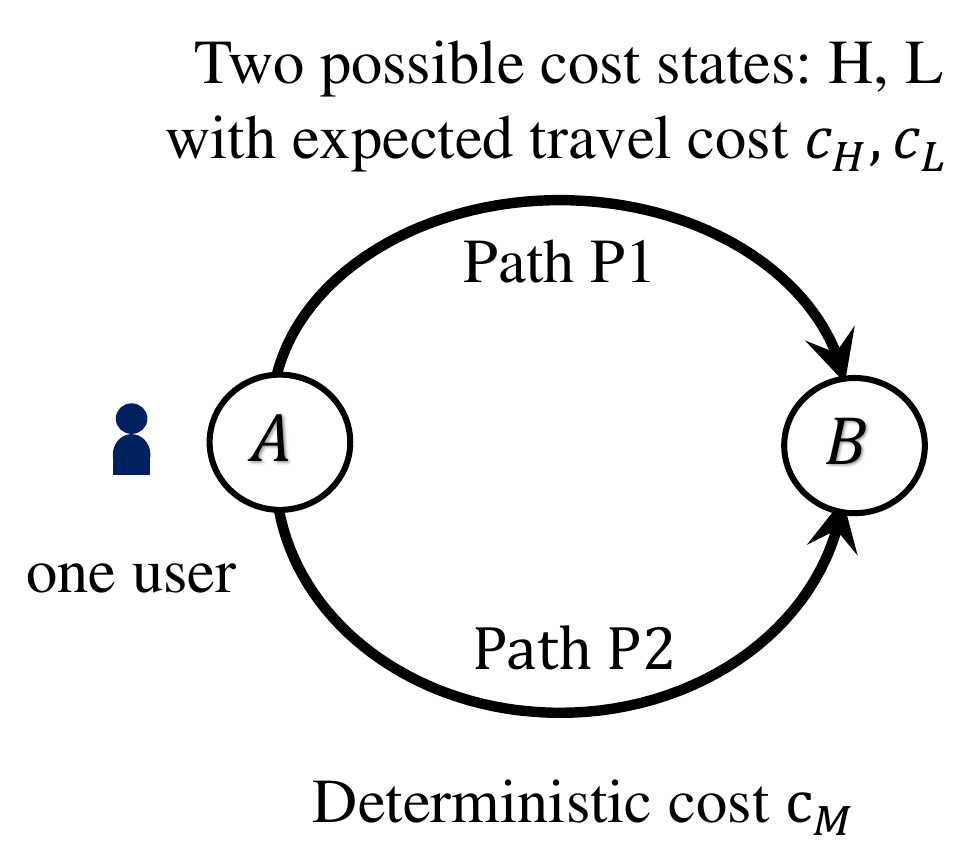}
\end{minipage}
}
\subfigure[The Markov chain for P1]{
\begin{minipage}[b]{0.45\linewidth}\label{ts}
 \centering
  \includegraphics[width=3in]{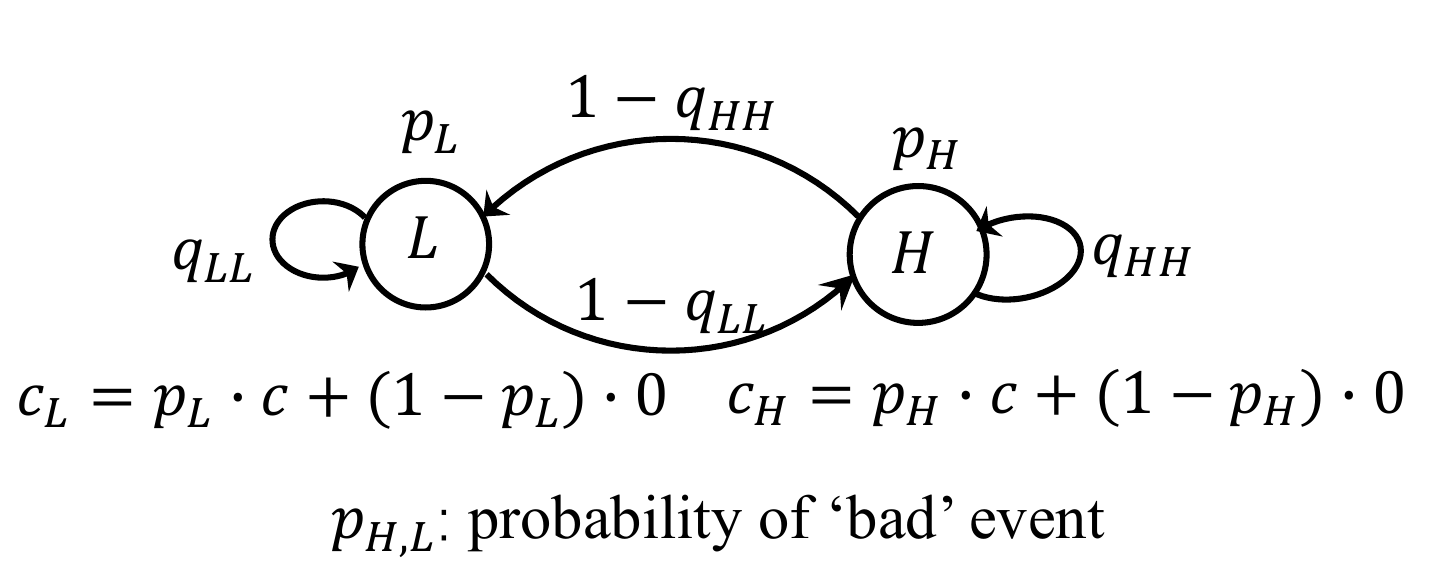}
\end{minipage}
}
\vspace{-5pt}
  \caption{Users that travel on P1 incur cost $c$ with probability $p_H$, $p_L$ ($p_H>p_L$) that depends on the cost sate of the path, which alternates between $H$ and $L$ according to the two-state Markov chain in Fig. \ref{ts} .
  }
  \label{tpts}\vspace{-10pt}
\end{figure*}
As mentioned in the Introduction, we model selfish behaviour of platform users. To make the problem
non-trivial we consider the challenging case that such users are `infinitely sophisticated' in terms of analytical and computational
capabilities and have full information about the system parameters and the platform algorithms.
We like to investigate the actions of such users and the corresponding results in social cost if
i) there is no platform recommending an action,
ii) the platform besides recommending a path is also
making available the full information collected so far by other users, and iii)
such information is hidden and only the current path recommendation is available.
To make the above problem well defined we use as a benchmark the case of
an optimal platform and then analyse what happens in the practical case of a platform that
uses machine learning, in particular using the Q-learning algorithm.

Our optimal platform makes routing decisions under uncertainty
capturing the fundamental tradeoff between exploring new possibilities versus exploiting optimally the current
information. To make the problem analytically tractable, we choose a network model that is simple
but fundamental enough to capture the essential aspects of making such routing decisions.


\subsection{Network Model}

We consider the simplest case where there are only two paths for our users to choose from:
one with deterministic cost
and another that alternates randomly between two states,
each such state generating a different average cost.
A platform user that travels along the stochastic path probes the path and experiences some actual cost which is
reported to the platform.
The platform collects these cost reports into a path history
and uses Bayesian inference  to determine the  probability that the path is in high or low cost state.
Though simple, this two-path network model captures the fundamental exploration-exploitation tradeoffs in
making routing decisions,
and makes users face the
incentive problems we like to analyse.\footnote{Note that our analysis can be easily extended to include multiple paths with deterministic costs in a larger network, by removing all the deterministic paths apart from the one with the smallest cost for routing consideration. Yet the analysis for multiple paths with time-varying costs is more involved and we need to update and balance the belief states of all stochastic paths. Still, Section \ref{multipath} provides some interesting results for developing the optimal threshold-based policy and examining the incentive compatibility for users to follow the platform recommendations.}

Our (road) network model with source node $A$ and destination node $B$ is in Fig. \ref{tpts}(a),
with two paths from the set \{P1, P2\}.
We consider an infinite discrete time horizon $t=1,2,\ldots$, and assume that during each discrete epoch there is
a single user
of our platform that must travel from $A$ to $B$ and must choose between paths P1 and P2.
In this abstract model a trip takes a single epoch to complete\footnote{We can easily extend the model where a trip
takes any fixed number of epochs.}.

We define the road condition experienced by a traveler on path P1 as a binary random variable $Y$:
\begin{itemize}
\item $Y=1$ is the event that a hazard occurs to the traveler
(e.g., poor visibility, `black ice' segments, congestion), i.e., driving on the path
generates some positive fixed driving cost $c$.
\item $Y=0$ is the event that no hazard occurs to the traveler;
without loss of generality we associate with this case a zero driving cost.
\end{itemize}
Users that drive on P1 observe the value of $Y$, and incur the corresponding cost depending whether $Y=1$ or $Y=0$.

To capture the randomness of the road condition of P1, we assume that  P1 alternates between two states $H$ and $L$ during $t=1,2,\ldots$
as a Markov chain with transition probabilities as in Fig. \ref{tpts}(b), and in each such state $Y$ is i.i.d. with a different distribution.
In state $H$ the probability of incurring a hazard $P[Y=1| H]=p_H$,
whereas in state $L$ this probability is $P[Y=1|L]=p_L$, where $p_L<p_H$.
Path P2 is always in a known cost state, generating cost $c_M$ such that $0\leq c_M\leq c$.\footnote{Otherwise, P2 will never be chosen due to its always higher cost than P1.}
Since $p_L<p_H$, $H$ corresponds to the high (expected) travel cost state, with average cost per traveller $c_H=p_H \cdot c$.
Similarly, $L$ is the low travel cost state with average cost $c_L=p_L\cdot c$.
Note that if P1 is in the high cost state, there is always some probability $1-p_H$ that a traveller incurs no hazard.
Similarly, if P1 is in the low cost state, there is still some probability $p_L$ that a traveller incurs a hazard.

A user that travels on P1 observes $Y$.
If $Y=1(0)$ we say that her observation is $1(0)$.
A user travelling along P2 observes nothing about the condition of path P1,
in which case we say her observation is $\emptyset$ (provides no information about P1 due to travel on P2).
A user  always shares  her observation about $Y$ with
the platform.
We denote the observation of a user that traveled at time $t$ by $y_t$, where $y_t\in\{0,1,\emptyset\}$.
The history of observations available to the platform by time $t$ corresponds to $(y_1,y_2,\ldots,y_t)$.

%
%

\subsection{Platform Information Model}\label{platform}

We next introduce how the platform works.
Given the history of observations $(y_1,y_2,\ldots,y_t)$,
it determines the probability that the path is in state $H$ or $L$ using Bayesian inferencing.
To avoid keeping an ever-increasing history of observations, we summarize the available
information equivalently into a single belief state $x_t$,
the probability that path P1 is in state $H$ \emph{just before} the travel of the user at time $t$. We denote  the platform's  initial belief state as $x_1$.

To make our Bayesian inferencing precise, we need  to define in our  model
our refined sequence of events from $t$ to $t+1$.
To do that we refine time and use $t^-,t,t^+$ as `micro' time refinements around time $t$ (where $t^+ < (t+1)^-$).
\begin{itemize}
\item At time $t^-$ there is no event occurring;
we just summarise our belief about P1's state based on the previous history:
compute the prior probability $x_t$, i.e., the probability for P1 being in $H$ just before $t$.
\item At time $t$ a user probes the paths by traveling and she supplies her trip observation $y_t$.
We use $y_t$ to update our posterior probability $x^\prime_t$ for the state of P1 being $H$ at time $t$
after the trip observation.
\item At time $t^+$ the Markov chain of the path state makes a transition.
\end{itemize}
In this model we consider that road conditions in P1 change in time scales slower or equal to the time scale of user
trip arrivals.
Then two consecutive users do not see P1 in its steady state distribution, and hence
the probability for $Y=1$ depends on the history of the observations.

The  belief state $x_{t+1}$ can be derived in a recursive way from the observations $y_t$ and $x_t$.
Let $a_t\in\{1,2\}$ be the choice of path of the user
who travels at time $t$. Consider first that $a_t=1$.
If $y_t=0$, then by Bayes' Theorem, the posterior probability that the cost state is $H$ after time $t$ is
 \begin{equation}\label{bayes0}
 \begin{split}
x^\prime_t&=\Pr[H|y_t=0,x_t]=\frac{\Pr[H,y_t=0,x_t]}{\Pr[H,y_t=0,x_t]+\Pr[L,y_t=0,x_t]}\\
&=\frac{x_t(1-p_H)}{x_t(1-p_H)+(1-x_t)(1-p_L)}\,,
\end{split}
\end{equation}
 where we use the fact that the path state does not change during $t^-,t$.
 Similarly, if $y_t=1$, we obtain
\begin{equation}\label{bayes1}
x^\prime_t=
\frac{x_t p_H}{x_tp_H+(1-x_t)p_L}.
\end{equation}
If $a_t=2$, then $y_t=\emptyset$ and the posterior probability is the same as the prior probability, i.e.,  $x^\prime_t=x_t$.

Given the posterior probability $x^\prime_t$, we can finally compute the probability that P1 is in state $H$ at $(t+1)^-$ as
   \begin{equation}\label{stateupdate}
x_{t+1}=x^\prime_t q_{HH}+(1-x^\prime_t)(1-q_{LL}).
 \end{equation}
Observe that a user offering positive information to the platform by travelling on P1 incurs an average cost of
\begin{equation}\label{costp1}
(x_tp_H+(1-x_t)p_L)c=x_tc_H+(1-x_t)c_L,
 \end{equation}
which might be more than the safe travel on P2 with fixed cost $c_M$.
This creates a tension between
individual incentives and social optimality as we analyse next in the optimal platform design problem.
%
%

 \section{The Optimal Platform by Solving POMDP}\label{optimal}
The optimal platform operation is modelled as
a Markov decision process (MDP) where the state is our belief state $x_t$,  decisions
correspond to path choices for travelers, and the cost function is the total discounted cost from travel.
In fact, our problem can be seen as a partially observable Markov decision process (POMDP),
and it is a standard solution method to reformulate it as an MDP over a belief state.
Though this optimal design problem is notoriously difficult to solve,
it provides a performance upper bound to evaluate i) myopic platforms and ii) model-free machine learning platforms.

 A stationary routing policy is a function $\pi$ that specifies an action $\pi(x)$ for each state $x$ at any time.
 Given the initial belief $x_1$,  the goal of the optimal platform is to find an optimal stationary policy
 $\pi$ to minimize the  expected total discounted driving cost (social cost) over an infinite time horizon, i.e.,
 \begin{equation}\label{objective}
 \small
\min\limits_{\pi}V(x,\pi)=\min\limits_{\pi}\lim_{\tau\rightarrow\infty}E_\pi\left[\sum\limits_{t=1}^{\tau}\beta^{t-1}C(x_t,a_t)|x_1=x\right],
 \end{equation}
 where $0<\beta<1$ is the discount factor over time and $C(x_t,a_t)$ is either \eqref{costp1} or $c_M$ if the specified routing action $\pi(x_t)$   is $1$ or $2$.
 We refer to the minimum cost value solution of the Bellman equation \eqref{objective}
 as the `value function'.
According to our discussion of  belief state updating in Section \ref{platform}, the specific optimality equation of our problem can be written as follows:
{
 \begin{align}\label{DP}
V(x)=&\min\{xc_H+(1-x)c_L+\beta(xp_H+(1-x)p_L)\cdot\notag\\
&V\left(\frac{xp_Hq_{HH}+(1-x)p_L(1-q_{LL})}{xp_H+(1-x)p_L}\right)+\nonumber\\
 &\beta(x(1-p_H)+(1-x)(1-p_L))\cdot\notag\\
 &V\left(\frac{x(1-p_H)q_{HH}+(1-x)(1-p_L)(1-q_{LL})}{x(1-p_H)+(1-x)(1-p_L)}\right),\nonumber\\
 &c_M+\beta V(xq_{HH}+(1-x)(1-q_{LL}))\}\nonumber\\
=&\min\{Q(x,1),Q(x,2)\}.
 \end{align}}
For ease of reading, we denote by  $Q(x,1)$ and $Q(x,2)$ the first and second terms in the minimum operator of \eqref{DP}, respectively. Hence, $Q(x,a)$ is the expected discounted cost staring from state $x$ if action $a$ is taken at the first time epoch and optimal policy is followed thereafter. Once we determine the exact value function, the optimal policy $\pi_{opt}$ can be obtained for any state $x$ as,
  \begin{equation}\label{optpolicy}
\pi_{opt}(x):=\arg\min_{a\in\{1,2\}}Q(x,a).
 \end{equation}
We can easily show that the optimal platform might recommend users to travel to P1
even when the expected travel cost in (\ref{costp1}) is higher than $c_M$ of P2 (i.e., when the myopic decision is P2) for exploration benefit in the future.

Although our analysis of the above POMDP and the corresponding incentive issues
is possible for any set of parameters, to illustrate better our key ideas and results we choose a specific
set of parameters as follows.
 \begin{assumption}\label{assum}
The Markov chain in Fig. \ref{ts} is symmetric with $q_{HH}=q_{LL}=q$ where $q\in[1/2,1)$, and the probabilities $p_H$ and $p_L$ are complementary, i.e., $p_H=p$ and $p_L=1-p$ where $p\in[1/2,1]$.
\end{assumption}
In the rest of the paper, we assume that Assumption \ref{assum} holds.
Without it, the more general problem can still be analysed in a similar way and yields the same theoretical results.

Before solving \eqref{DP} we can first prove it has a unique solution by using the contraction mapping theorem.
Note that the minimum operator in \eqref{DP} is a contraction operator since $\beta<1$.
Furthermore, we prove that the value function $V(x)$ is a piecewise-linear concave function of the belief state $x$ by mathematical induction.
Besides, we show that $V(x)$ is an increasing function of $x$.  Here we skip detailed proofs due to page limit.
 \begin{proposition}\label{property}
There exists a unique solution to the optimality equation \eqref{DP} and it is a piecewise-linear, increasing and concave function of the belief state $x$.
\end{proposition}

The proof is given in Appendix \ref{propertyproof}.
Although the existence of solution to \eqref{DP} is guaranteed, it is still difficult to solve it analytically.
An intuitive conjecture\footnote{
Our POMDP model is similar (but not the same) to the well studied problem of searching for a moving object \cite{ross1983introduction}. To prove the same conjecture for that problem still remains an open problem.
This suggests that
proving (or disproving) the threshold property for the optimal policy in our case can be extremely challenging.
Yet using extensive numerical analysis for a very fine grid of parameter values we have observed that Conjecture \ref{conjecture} remains true.
}
 about the optimal policy is that it is of threshold type.
\begin{conjecture}\label{conjecture}
There exists a threshold value $x^*\in[0,1]$ such that it is optimal to choose path P1 when the belief state is in $[0,x^*)$, choose P2 when the belief state is in $(x^*,1]$, and choose any of the two paths at $x^*$.
\end{conjecture}

We have been able to formally prove our conjecture for a restricted set of parameters as follows, by using the concavity of the value function $V(x)$.
\begin{proposition}\label{conj}
If $\beta(2q-1)<2/3$,  the optimal policy is of threshold type.
\end{proposition}
The proof is given in Appendix \ref{conjproof}.
We have the following corollary which directly follows from Proposition \ref{conj}.
\begin{coro}
If $q\in[1/2,5/6]$ or $\beta\in[0,2/3)$, the optimal policy is of threshold type.
\end{coro}
For our experimental analysis we will discretise finely the state space $[0,1]$, use value iteration to compute the
value function in \eqref{DP}, and finally compute the optimal policy at each
given belief state by solving \eqref{optpolicy}.

How much would society lose compared to the optimum
if a myopic platform (always chooses the least current cost path) is in place?
We see in the next section that this performance loss can be arbitrarily large.

\section{Myopic Platforms and PoA}\label{myopic}
Myopic platforms such as Waze and Google Maps estimate the travel cost of different paths and suggest to
users the path with the smallest cost.
In this section we introduce two basic myopic platforms: a platform that does not use feedback information from users,
and a platform that uses such feedback to update the current cost estimate.
We analyse these platforms and characterise
their performance gaps with the optimal platform in Section \ref{optimal}  in term of price of anarchy (PoA).
The large PoA values resulting from our analysis suggest that  the optimal platform is definitely desired, but such platform is not
incentive compatible. This motivates our incentive alignment proposal in the rest of the paper.

\subsection{Myopic Platform without Information Sharing}\label{myopicwo}
In this case the platform uses long-run average path costs to make recommendations.
For our specific model parameters, cost states $H$ and $L$ have each probability $1/2$
and  the expected cost to travel through path P1 is $(c_H+c_L)/2$.
The routing policy of the myopic platform is straightforward. Let $\pi_{\emptyset}$ denote the routing policy
without information sharing, then
\begin{equation}
\label{myopicnoinfo}
\pi_{\emptyset}(x) =
\begin{cases} 1&\mbox{ if } c_M\ge (c_H+c_L)/2,\\
2& \mbox{ if }c_M<(c_H+c_L)/2,
\end{cases}\nonumber
\end{equation}
which is independent of $x$. Thus, $\pi_\emptyset$ either chooses path P1 all the time or path P2.
We can now calculate the value function.
If $c_M<(c_H+c_L)/2$, $\pi_\emptyset$ always chooses path P2 to incur immediate cost $c_M$ to users over time, we have
$$V_{\pi_\emptyset}(x)=\frac{c_M}{1-\beta}.$$

If $c_M\ge (c_H+c_L)/2$, $\pi_\emptyset$ always chooses P1.
Given some initial probability $x$ about path P1 (assumed known to the platform), the value function satisfies
$$V_{\pi_\emptyset}(x)=xc_H+(1-x)c_L+\beta V_{\pi_\emptyset}(xq+(1-x)(1-q)).$$
Similar to the proof of Proposition \ref{property}, we can prove the existence and uniqueness of
$V_{\pi_\emptyset}(x)$. We can also prove by mathematical induction that $V_{\pi_\emptyset}(x)$ is a linear function of $x$. It follows that,
$$V_{\pi_\emptyset}(x)=\frac{c_H-c_L}{1+\beta-2q\beta}x+\frac{(\beta-q\beta)c_H+(1-q\beta)c_L}{(1+\beta-2q\beta)(1-\beta)}.$$

$PoA>1$ is defined as the ratio between the maximum expected total discounted cost  incurred under this myopic policy $\pi_\emptyset$ and the minimum expected total discounted cost $V(x)$ in \eqref{DP}, by searching over all possible network parameters. That is,
\begin{equation}\label{poadef}
PoA_{\pi_\emptyset}=\max\limits_{p,q,c,c_M,x}\frac{V_{\pi_\emptyset}(x)}{V(x)}.
\end{equation}

\begin{proposition}\label{poapie}
Given $\beta<1$ and $c_M>0$, the  policy $\pi_\emptyset$ achieves an infinite price of anarchy, i.e.,
$PoA_{\pi_\emptyset}=\infty$.
\end{proposition}
\emph{Sketch of Proof:} Lets rescale costs so that $c_L=0$. To determine the PoA, we purposely create a worse case scenario where  $\pi_\emptyset$ always chooses path P2 ($p=1$ and $c>2c_M$). Furthermore, let the initial P1 state be $L$ (i.e., $x=0$) and let the Markov chain change very slowly ($q\rightarrow1$). Then path P1 will remain in $L$ for a very long time. Since $\pi_\emptyset$ always chooses path P2, its cost value is a constant $\frac{c_M}{1-\beta}$. Since $1-p=0$ there is zero average cost in state $L$  and the Markov chain is fully observable; hence the optimal policy will choose path P1 until a change of state occurs, i.e., non-zero cost is observed. But the time of such a transition can be made arbitrarily large since $q\rightarrow1$ while our cost discount factor remains constant and equal to $\beta$. A more formal argument in Appendix \ref{poapieproof} can be used to prove that the price of anarchy of $\pi_\emptyset$ is infinity.

To prove Proposition \ref{poapie}, we can purposely create the worst case scenario with properly chosen initial state $x$ and costs $c$ and $c_M$,
 where  $\pi_\emptyset$ always chooses path P2 but  the optimal policy chooses path P1 until a non-zero cost is observed. In this case, the expected cost of optimal policy can be made arbitrarily close to zero.

 Even though $\pi_\emptyset$ can be arbitrarily worse
than the optimal policy, users will still follow the platform recommendation under $\pi_\emptyset$.
Without any other information, sophisticated users can reproduce the calculations of the platform and hence
will follow $\pi_\emptyset$.

\subsection{Myopic Platform with Information Sharing}
We now consider a myopic platform where travelers share information online.
The difference from the optimal platform is that here it
chooses actions that myopically  minimise immediate average costs.
Given the current belief $x$ about P1, the immediate expected cost  is $xc_H+(1-x)c_L$ for path P1 and $c_M$ for path P2. By equating the two costs and solving for the corresponding threshold belief state $\hat{x}$,
we obtain $\hat{x}=\frac{c_M-c_L}{c_H-c_L}$. The myopic policy of this platform is
\begin{equation}
\label{myopicwithinfo}
\pi_{m}(x) =
\begin{cases} 1&\mbox{ if } x\le\hat{x},\\
2& \mbox{ if }x>\hat{x}.
\end{cases}\nonumber
\end{equation}
Note that users will follow the recommendation of the platform as their objectives are aligned.

Let $V_{\pi_m}(x)$ be the cost value function under the myopic policy $\pi_m$. Similar to (5), we obtain
{\small\begin{align}
\label{eq:my}
&V_{\pi_m}(x)
=\begin{cases}
xc_H+(1-x)c_L+\beta\big(xp+(1-x)(1-p)\big)\cdot\\
 V_{\pi_m}(\frac{xpq+(1-x)(1-p)(1-q)}{xp+(1-x)(1-p)})+\beta\big((1-x)p+x(1-p)\big) \\ \cdot V_{\pi_m}(\frac{x(1-p)q+(1-x)p(1-q)}{x(1-p)+(1-x)p}),\hfill \ \ \qquad\mbox{ if }0\le x\le \hat{x},\hfill\\
c_M+\beta V_{\pi_m}(xq+(1-x)(1-q))\hfill\mbox{ if }\hat{x}< x\le1.\hfill
\end{cases}
\end{align}}
It is rather obvious that this myopic platform behaves the same way as the Nash equilibrium of a
system that deploys the optimum platform but users have full information about travel history, i.e.,
can reconstruct $x$. This is because in the optimal platform users will still use \eqref{myopicwithinfo} to choose paths,
and hence the two systems will have the same sample paths on our probability space.
\begin{fact}
On any sample path, the Nash equilibrium of the optimal platform with information sharing and selfish users is the same
as the Nash equilibrium of the myopic platform (with information sharing and selfish users) using $\pi_m$.
\end{fact}

One can easily prove the intuitive result that $\pi_m$ is more conservative than $\pi_{opt}$
in the sense that
if  $\pi_m$ prefers the risky path P1, then clearly $\pi_{opt}$ should also prefer it since it obtains the
additional/future benefit
of learning the path state more accurately. Obviously, the reverse does not hold:
if $\pi_{opt}$ prefers P1, it does not imply that $\pi_m$ should also prefer P1.
This is formally stated in the following proposition and will be used in our proof for
incentive compatibility in Section \ref{mechanism}.
\begin{proposition}\label{less}
For any $x\in[0,\hat{x}]$ the optimal policy $\pi_{opt}$ chooses path P1.
\end{proposition}
\begin{proof}
Note  that when $x\in[0,\hat{x}]$,
$$xc_H+(1-x)c_L\le\hat{x}c_H+(1-\hat{x})c_L=c_M.$$
By the concavity of the value function $V(x)$,
 \begin{align}
&\beta(xp_H+(1-x)p_L)\cdot V\left(\frac{xp_Hq_{HH}+(1-x)p_L(1-q_{LL})}{xp_H+(1-x)p_L}\right)\nonumber\\
&+\beta(x(1-p_H)+(1-x)(1-p_L))\cdot \notag\\
 &V\left(\frac{x(1-p_H)q_{HH}+(1-x)(1-p_L)(1-q_{LL})}{x(1-p_H)+(1-x)(1-p_L)}\right)\nonumber\\
&\le\beta V(xq_{HH}+(1-x)(1-q_{LL})).\nonumber
\end{align}
By combining the above two inequalities, we obtain $Q(x,1)\le Q(x,2)$ when $x\in[0,\hat{x}]$. This completes the proof.
\end{proof}
 Note that if Conjecture \ref{conjecture} is true, a corollary is that $x^*>\hat{x}$.


Similar to \eqref{poadef} the price of anarchy of $\pi_m$ is defined as
\[
PoA_{\pi_m}=\max\limits_{p,q,c,c_M,x}\frac{V_{\pi_m}(x)}{V(x)}.\]
\begin{proposition}\label{poapim}
Given $\beta<1$ and $c_M>0$, the  policy $\pi_m$ achieves  $PoA_{\pi_m}=\frac{1}{1-\beta}$.
\end{proposition}

 \emph{Sketch of Proof:}  Let's rescale costs so that $c_L=0$. Let the Markov chain be fully observable (i.e., $p=1$), and let it change very slowly (i.e., $q\rightarrow1$). Let the initial probability $x>0$ be very small. Thus, with a very high probability, path P1 starts in state $L$ and remain in that state for very long time thereafter. Now choose $c_M$ slightly smaller than $xc_H+(1-x)c_L=xc_H$ so that $\pi_m$ chooses path P2 at the beginning. Without exploring path P1, the belief state $x$ will gradually increase with time and in turn $\pi_m$ continues choosing path P2 instead of exploring path P1. Hence,  policy $\pi_m$ will always choose path P2 generating cost $c_M$ in every time epoch. But the optimal policy would like to take a little risk exploring path P1 at the beginning to exclude the possibility that it is in state $H$ (which is highly improbable) to keep exploiting the zero cost of state $L$ if this turns out to be the case. If the cost state turns out to be $H$ (which occurs with very low probability), we switch to path P2 thereafter imitating $\pi_m$. Hence exploring path P1 at the beginning  generates a cost of $xc_H \approx c_M$, but from the second time epoch and for a very long time forward the cost under the optimal policy is either always $c_M$ (with prob. $x$) or $c_L=0$ (with prob $1-x$).Simple calculations give the result as $x\rightarrow0$. The detailed proof can be found in Appendix \ref{poapimproof}.

Similar to the proof idea of Proposition \ref{poapie}, we still purposely create the worst case, where  $\pi_m$ always chooses path P2 but  the optimal policy chooses path P1 until a non-zero cost is observed. But with information sharing, we cannot make the expected cost of the optimal policy arbitrarily close to zero.
 Thus, unlike $\pi_{\emptyset}$, PoA of $\pi_m$ is bounded.
This is because obtaining information from travelers allows the platform to significantly reduce the immediate cost.
Without such information, the platform can make terrible routing decision from the start.
However, even with information sharing, the decision making of the platform can still be arbitrarily poor in the long term.
The performance of the myopic platform becomes worse compared to the optimal policy as the discount factor $\beta$ increases and future costs become more important.
As $\beta$ approaches 1, PoA approaches infinity, indicating a great performance loss due to the myopic nature of $\pi_m$.
As this performance loss can be huge, it is crucial to design incentive mechanisms for $\pi_{opt}$ for achieving incentive compatibility.

\section{Information Restriction Mechanism for Incentive Compatibility of $\pi_{opt}$}\label{mechanism}
To provide incentives for users to follow the recommendations of the optimal platform,
we propose a novel information restriction mechanism.
The idea is to hide from users the information collected by the platform from the previous travelers
and supply only the path recommendation.
This is equivalent to keep private the information about the current value of the
belief state $x$ that the optimal platform has constructed.
Hence, a user knows only her current path recommendation besides knowing the
statistical properties of the paths and the platform algorithm.

We use the concept of correlated equilibrium (proposed by Robert Aumann \cite{aumann1974subjectivity}).
In this model the platform provides a private signal to the players
which then act in their best interest under information uncertainty.
In our case the platform offers a private signal (its recommendation) and users decide to follow it or not.
If no user would want to deviate from the recommendation assuming the others don't deviate,
we say all users following recommendations is  a correlated equilibrium.
The mechanism we propose here does not require the optimal policy to be of threshold type (Conjecture \ref{conjecture}),
and its incentive properties are just related to properties of the value function of the optimal policy.

The optimal policy always produces a partition of the belief state space
$\mathcal{X}=[0,1]$ into two sets $\mathcal{X}_1$, $\mathcal{X}_2$,
where $\mathcal{X}_a$, $a=1,2$ is the set of belief states for which the optimal policy $\pi_{opt}$ chooses action $a$.
Our signalling mechanism is defined as follows.
\begin{definition}
Information Restriction Mechanism (IRM): The platform hides  the history of observations
(hence the belief state information $x$) from the users.
It follows $\pi_{opt}$ in \eqref{optpolicy} and recommends P1 when the belief state $x\in\mathcal{X}_1$ and
P2 when $x\in\mathcal{X}_2$.
\end{definition}

IRM is incentive compatible if no user wants to deviate from her path recommendation unilaterally.

Next, we analyse the users' actions (to follow the recommendation or not) in the correlated equilibrium under this mechanism. Although users have no knowledge of $x$ in real time, they are aware of the actual Markov chain model of the paths,
the value of the parameters and the algorithm of the platform.
They will reverse-engineer the platform recommendation to estimate the possible values of the actual belief state
and based on that decide on following the recommendation or not.
More specifically, when the recommendation is P2,
the user will infer that  the current system state $x$ must be in $\mathcal{X}_2$,
which implies that $x\ge\hat{x}$ by Proposition \ref{less}.
Note that the user benefits from choosing P1 for $x\leq \hat{x}$ and P2 for $x>\hat{x}$. Thus, the user will follow the recommendation of P2.
When the recommendation is P1, the user infers that the current system state $x$ must be in $\mathcal{X}_1$.
We can prove that in the average sense she benefits by choosing P1 assuming the rest of the users do the same,
and hence she will follow the recommendation of IRM. The incentive compatibility and the efficiency of IRM are formally stated in the
next theorem.
\begin{theorem}\label{theorem}
Under IRM, all users following the optimal platform's recommendation is a correlated equilibrium.
Thus, our IRM achieves optimality and  $PoA=1$ .
\end{theorem}
\begin{proof}
Consider a user's point of view at time $t$ who assumes  that all the other users follow the optimal platform's recommendation. Lacking any information about the history of the path state and assuming that the system operates already for very long time and the rest of the users follow the recommendation of the platform, her best estimate of the belief state $x$ is the stationary distribution $P^{\pi_{opt}}(x)$ under the optimal policy $\pi_{opt}$ which then can be conditioned on the recommendation for P1 or P2. To prove our result we don't need to evaluate this distribution analytically, but we need to establish certain properties of $P^{\pi_{opt}}(x)$. To do that we use $P^{\pi_{opt}}(x)$ to evaluate the long-run \emph{un-discounted} average cost $\lambda_{\pi_{opt}}$ that the system would incur if the platform uses the \emph{discounted} cost optimal policy $\pi_{opt}$ and users follow it\footnote{Note that this is not the cost minimised by the platform and we only use it to establish a relation involving $P^{\pi_{opt}}(x)$ to be used later in the proof.}. Then
$\lambda_{\pi_{opt}}$ can be computed according to the stationary distribution $P^{\pi_{opt}}$.
 \begin{align}\label{claim}
\lambda_{\pi_{opt}}&=\int_{\mathcal{X}_1}(xc_H+(1-x)c_L)\dif P^{\pi_{opt}}(x)+\int_{\mathcal{X}_2}c_M\dif P^{\pi_{opt}}(x)\notag\\
&\le c_M,
\end{align}
where we used the claim that $\lambda_{\pi_{opt}}$ is  less than $c_M$ (to be proved later). The formula above simply states that when in $x\in \mathcal{X}_1$ the average cost of a user following the recommendations is $xc_H+(1-x)c_L$ and when $x\in \mathcal{X}_2$ this cost  is $c_M$.

When the recommendation is P2, the user can reverse engineer the recommendation to infer  that the current belief state $x$ must be in $\mathcal{X}_2$. According to  Proposition \ref{less}, whenever $x\in\mathcal{X}_2$, it follows that
\begin{equation}
\label{claim2}
xc_H+(1-x)c_L\ge c_M.
\end{equation}
 The user can compute the expected cost of  travelling along P1 when the recommendation is P2 according to the stationary distribution $P^{\pi_{opt}}$  of the belief state $x$. By \eqref{claim2}  it is larger than $c_M$, that is,
\begin{align}
E_{P^{\pi_{opt}}}&[xc_H+(1-x)c_L|x\in \mathcal{X}_2]\notag\\
=&\frac{\int_{\mathcal{X}_2}(xc_H+(1-x)c_L)\dif P^{\pi_{opt}}(x)}{\int_{\mathcal{X}_2}\dif P^{\pi_{opt}}(x)}\notag\\
\ge&\frac{\int_{\mathcal{X}_2}c_M\dif P^{\pi_{opt}}(x)}{\int_{\mathcal{X}_2}\dif P^{\pi_{opt}}(x)}=c_M.\nonumber
\end{align}
Hence, the platform user will follow the recommendation to choose P2. When the recommendation is P1, the user infers that the current system state $x$ must be in $\mathcal{X}_1$. She will compute the expected cost of  travelling along P1 according to $P^{\pi_{opt}}$. By using \eqref{claim}, this cost is smaller than $c_M$ since
\begin{align}
E_{P^{\pi_{opt}}}&[xc_H+(1-x)c_L|x\in \mathcal{X}_1]\notag\\
=&\frac{\int_{\mathcal{X}_1}(xc_H+(1-x)c_L)\dif P^{\pi_{opt}}(x)}{\int_{\mathcal{X}_1}\dif P^{\pi_{opt}}(x)}\notag\\
=&\frac{\lambda_{\pi_{opt}}-\int_{\mathcal{X}_2}c_M\dif P^{\pi_{opt}}(x)}{\int_{\mathcal{X}_1}\dif P^{\pi_{opt}}(x)}\notag\\
\le&\frac{c_M-\int_{\mathcal{X}_2}c_M\dif P^{\pi_{opt}}(x)}{\int_{\mathcal{X}_1}\dif P^{\pi_{opt}}(x)}=c_M.\notag
\end{align}
Hence, each myopic user will follow the recommendation to choose P1.

Now we still need to prove our claim that \eqref{claim} holds. Assume the initial distribution of the belief state is the stationary distribution $P^{\pi_{opt}}$. Then, the belief state $x_t$ at any time $t$ has the same probability distribution $P^{\pi_{opt}}$. Since policy $\pi_{opt}$ is optimal for the total discounted cost minimization problem, the resulting optimal expected total discounted cost averaged over the initial state distribution $P^{\pi_{opt}}$ is
\begin{align*}
\sum_{t=1}^\infty&\beta^{t-1}\bigg(\int_{\mathcal{X}_1}(xc_H+(1-x)c_L)\dif P^{\pi_{opt}}(x)+\notag\\
&\int_{\mathcal{X}_2}c_M\dif P^{\pi_{opt}}( x)\bigg)\notag\\
=&\sum_{t=1}^\infty\beta^{t-1}\lambda_{\pi_{opt}}=\frac{\lambda_{\pi_{opt}}}{1-\beta}.
\nonumber
\end{align*}
Note that one can always choose path P2 at each time epoch and the resulting expected total discounted cost is
$$\sum_{t=1}^\infty\beta^{t-1}c_M=\frac{c_M}{1-\beta},$$
which must be larger than the optimal cost. Thus,
$$\frac{\lambda_{\pi_{opt}}}{1-\beta}\le\frac{c_M}{1-\beta}$$
or $\lambda_{\pi_{opt}}\leq c_M$. Thus, \eqref{claim} holds and this completes the proof.
\end{proof}

\section{Reinforcement Learning Platform}\label{RL}
In practice, it may be difficult to develop an exact POMDP model for analysing the routing policy,
either because of the many unknown parameter values or because such a Markovian model may not be
sensible. Hence we expect that platforms will
resort to model-free reinforcement learning techniques  such as Q-learning \cite{watkins1989learning}.
We want  to obtain some insights on how reinforcement learning, which leads to sub-optimal platforms,
affects our mechanism results regarding incentive compatibility. In particular, we want  to
make the following conjecture which we have been able to test with experiments.
\begin{conjecture}
Under the IRM, as the machine learning algorithm becomes more efficient in reducing the average system cost,
the range of  system parameters for which the users follow the recommendations of the platform increases.
In particular, in the case of Q-learning algorithms (see \eqref{updaterule}), as $K\rightarrow \infty$, IRM induces IC.
\end{conjecture}
In simple terms, increasing platform efficiency combined with hiding information induces incentive compatibility in
a wider range of systems.
In this section we analyse the performance of such a learning platform
and measure its performance loss from the optimal platform
benchmark.

The classical Q-learning algorithm estimates the Q-value function in an online fashion and
computes the optimal policy according to Q-values computed for all possible system states and actions.
In this case, state $y$ records the latest $K$ observations (cost reports by the last $K$ travelers),
where $K$ is a parameter of the learning algorithm.
For each possible action $a$ in state $y$
the Q-value maps  the state-action tuple $(y,a)$ to the anticipated cost,
and the optimal action corresponding to  the minimum Q-value is chosen.
The platform updates the Q-values for each $(y,a)$  over time by learning from the path observations the actual costs that
such actions generate in the given context.

We expect the performance of Q-learning to improve as $K$ increases, since the system
makes decisions in a more detailed context. Another way to see this is that a larger $K$ allows for a better
estimate of the correct value of the belief state $x$ that the POMDP-based optimal platform would like to use
for its decisions. But larger values of $K$ come at an exponential increase of the size of the state-space $\mathcal{Y}$
(which is $3^K$ with each observation being $0/1/\emptyset$) and influence the time Q-learning needs to converge in its optimal choices.
Our numerical results later suggest that a small $K$ such as $K=3$ already provides near-optimal performance.
Next we describe the Q-learning algorithm adapted to our problem.

Given observation history $y_t$ before time $t$, the platform takes action $a_t$ and incurs actual cost $c^\prime_t$,
and updates the observation vector from $y_t$ to $y^\prime$. The Q-value is updated as:
{\small
\begin{align}\label{updaterule}
&Q_{t+1}(y,a)=\notag\\
&\quad\begin{cases}
\alpha_t(y,a)(c^\prime_t+\beta\min\limits_{a^\prime\in\{1,2\}} Q_t(y^\prime,a^\prime))\\
+(1-\alpha_t(y,a))Q_t(y,a)\qquad\mbox{if }y=y_t,\mbox{ and }a=a_t,\\
Q_t(y,a)\qquad\qquad\qquad\qquad\ \ \mbox{otherwise},
\end{cases}
\end{align}}where $\alpha_t$ is the learning rate. It is known from \cite{watkins1992q} that Q-learning converges if each $(y,a)$ tuple  is performed infinitely often and $\alpha_t(y,a)$ satisfies for each tuple $(y,a)$,
$$\sum\limits_{t=1}^\infty\alpha_t(y,a)=\infty \mbox{ and } \sum\limits_{t=1}^\infty\alpha_t^2(y,a)<\infty.$$
In our implementation we use $\alpha_t(y,a)=\frac{1}{(1+N(y,a,t))^\omega}$
where $N(y,a,t)$ is the number of times that the platform observes $y$ and performs $a$ until time
$t$ and $\omega\in(0.5,1]$ (as suggested in \cite{even2003learning}).
We next show that Q-learning has good performance when applied to the
benchmark POMDP path model
and then investigate incentive compatibility for users of this platform.

\subsection{Performance Analysis of Q-learning Platform}

In this subsection we provide a methodology for calculating the parameters of Q-learning
after it converges and hence solving the path selection policy obtained by Q-learning.
This allows us to compare this policy with the optimal policy $\pi_{opt}$ of the POMDP and formulate
the incentive compatibility problem faced by the users in Q-learning platform.

Using the results in \cite{singh1994learning} regarding the steady state values of the
parameters of the Q-learning algorithm, we obtain that
our Q-learning algorithm converges with probability 1 to the solution for each $y\in\mathcal{Y}$
of the following system of equations:
{\small
\begin{align}\label{Qlearning}
&Q(y,1)=\notag (\Pr[H|y]p+\Pr[L|y](1-p))(c+\beta\min_{a\in\{1,2\}}Q(y^\prime(1),a))\notag\\
&\quad+(\Pr[L|y]p+\Pr[H|y](1-p))(0+\beta\min_{a\in\{1,2\}}Q(y^\prime(0),a)),\notag\\
&Q(y,2)=c_M+\beta\min_{a\in\{1,2\}}Q(y^\prime(\emptyset),a).
\end{align}}Here, $y^\prime$ is the sequence of latest $K$ observations after the transition
by appending the last observation (0, 1, or $\emptyset$) to the vector $y$ after removing its first element.
$\Pr[H|y]$, $\Pr[L|y]$ are the asymptotic probabilities that the underlying cost state is $H$, $L$, respectively,
given that the sequence of $K$ latest observations is $y$.

\begin{figure}[!t]
  \centering\vspace{-5pt}
  \includegraphics[width=2.2in]{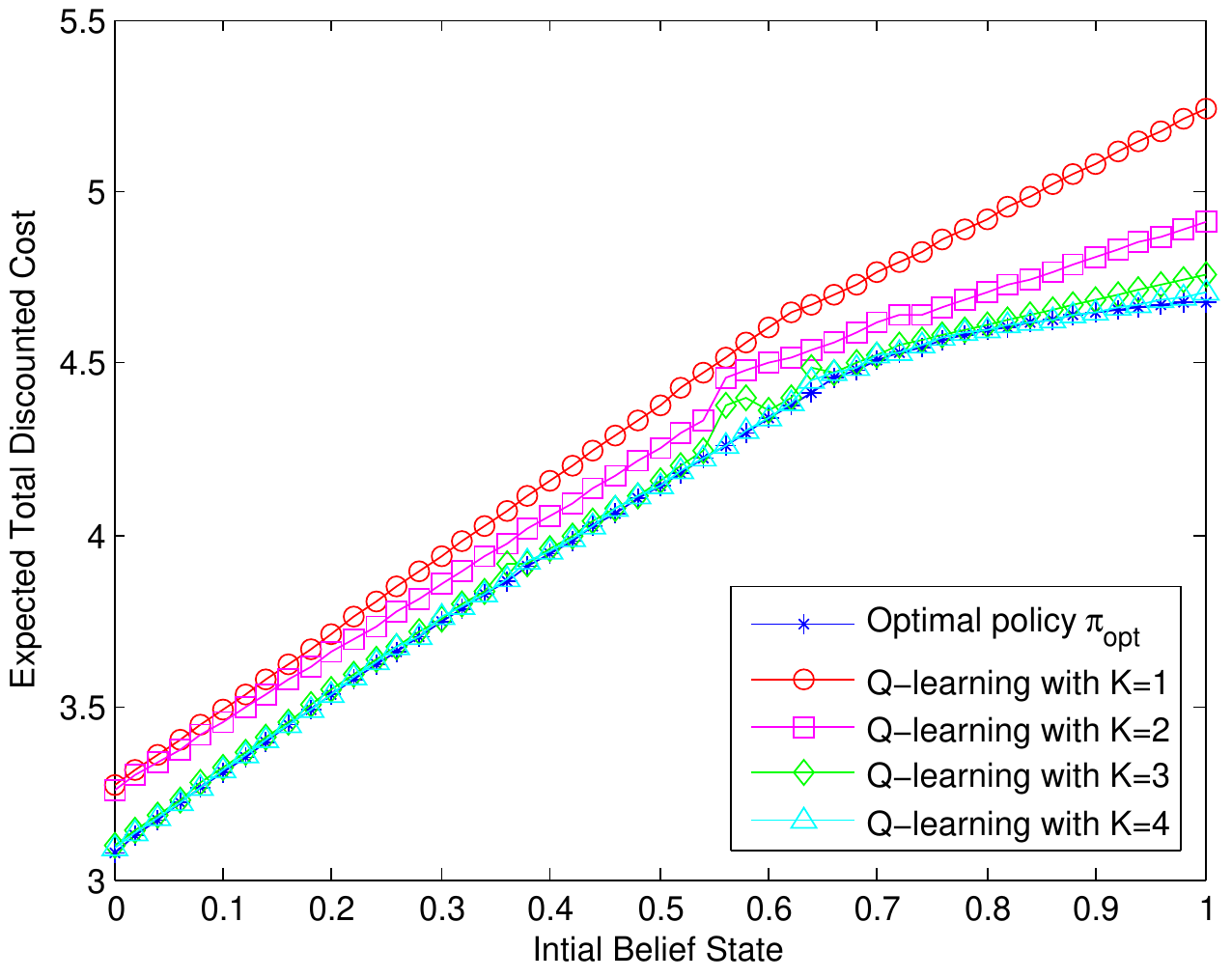}\vspace{-5pt}
  \caption{The expected total discounted cost of $\pi_{\bar{Q}}$ compared to the optimal total expected discounted cost for different values of $K$.  For small values of $K$, the Q-learning is suboptimal leading to higher values of cost compared to the optimal policy. As $K$ increases,
the expected total discounted cost of $\pi_{\bar{Q}}$  becomes closer to  the optimal total expected discounted cost. Here we set $\beta=p=q=0.9$, $c=1$, $c_M=0.5$. }
  \label{Qapprox}\vspace{-15pt}
\end{figure}


\begin{figure*}
\centering
\subfigure[$c_M$ regime]{
\begin{minipage}[b]{0.31\linewidth}\label{ICcm}
  \centering
  \includegraphics[width=2in]{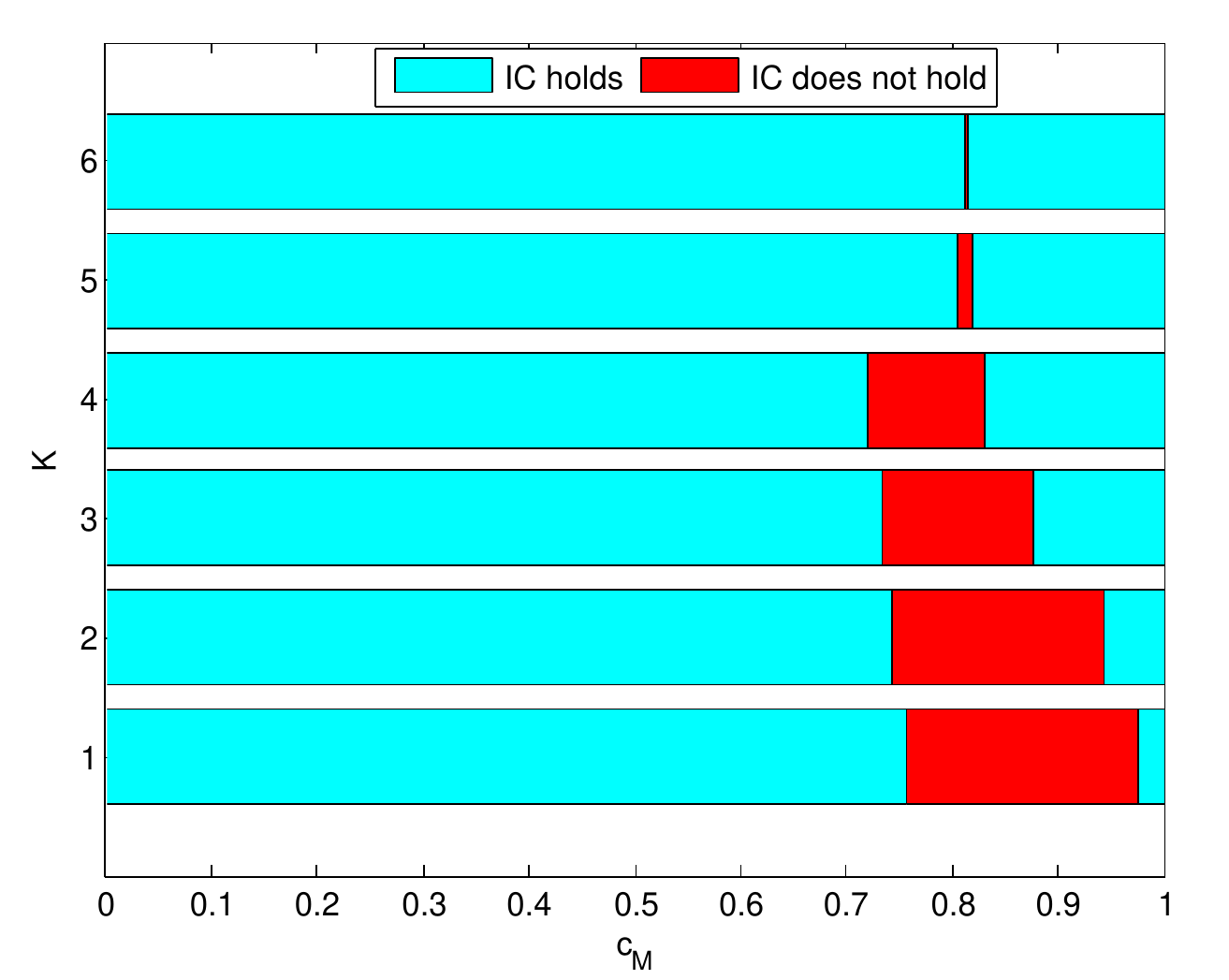}
\end{minipage}
}
\subfigure[$q$ regime]{
\begin{minipage}[b]{0.31\linewidth}\label{ICq}
 \centering
  \includegraphics[width=1.9in]{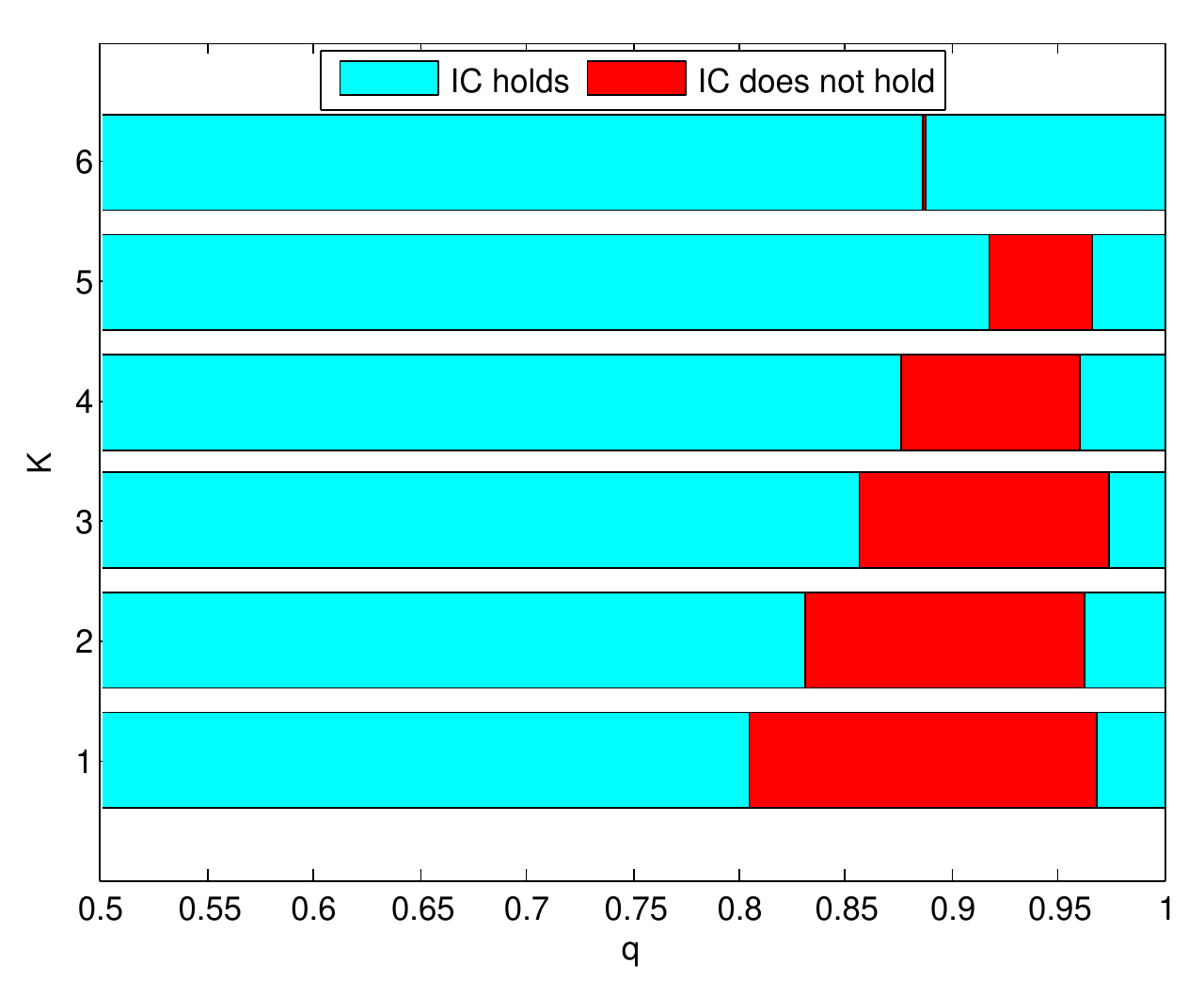}
\end{minipage}
}
\subfigure[$\beta$ regime]{
\begin{minipage}[b]{0.31\linewidth}\label{ICbeta}
  \centering
  \includegraphics[width=1.9in]{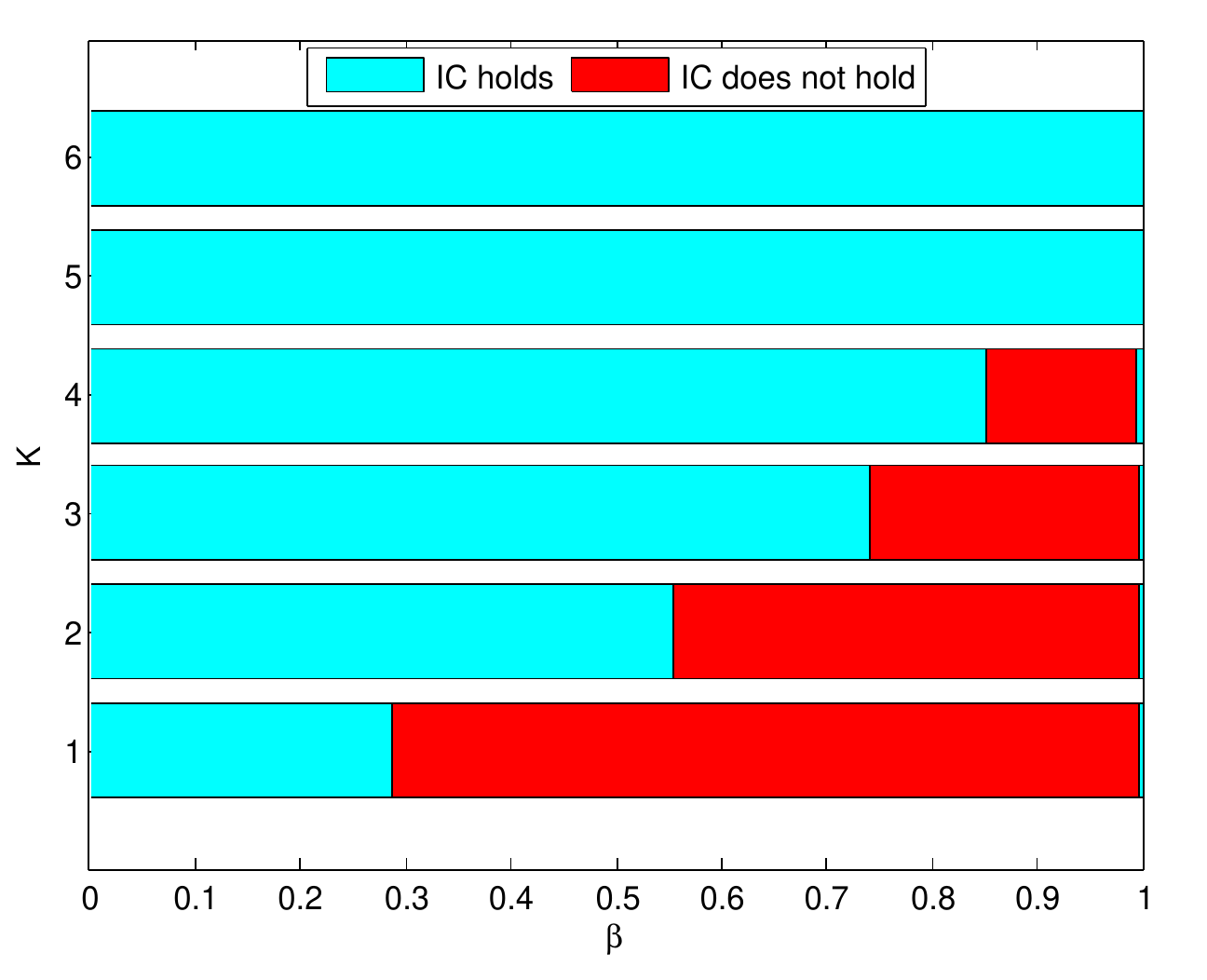}
\end{minipage}
}\vspace{-10pt}
  \caption{From left to right, we plot in red the regimes of (a) all possible $c_M$ and $K$ values, (b) all possible $q$ and $K$ values and (c) all possible $p$ and $K$ values for which incentive compatibility does not hold. In general we see that IC holds for large range of instances and as $K$ increases, the interval of regime in which the IC does not hold becomes smaller. In Fig. \ref{ICcm}, we set $\beta=p=q=0.9$, $c=1$, and $c_M=0,0.001,0.002,\cdots,0.999,1$.  In Fig. \ref{ICq}, we set $\beta=p=0.9$, $c_M=0.8$, $c=1$, and $q=0.5,0.501,\cdots,0.999$. In Fig. \ref{ICbeta}, we set $p=q=0.9$, $c_M=0.8$, $c=1$, and $\beta=0.001,0.002,\cdots,0.999$.}
  \label{IC}\vspace{-10pt}
\end{figure*}

We can use \eqref{bayes1} and \eqref{stateupdate} to compute $x_{t+1}$ from some initial state $x_1$ (assumed $1/2$
in our specific case or equal to the steady state distribution of the path Markov chain in general)
and any sequence of observations $y=(y_1,\ldots,y_t)$. This defines $\Pr[H|y]$ (and $\Pr[L|y]=1-\Pr[H|y]$)
in \eqref{Qlearning} for all possible values $y\in\mathcal{Y}$.

Let $\bar{Q}(y,a)$ be the solution to equation \eqref{Qlearning} with corresponding (asymptotic) policy $\pi_{\bar{Q}}$.
This takes the action with the minimum Q-value, i.e.,
\begin{equation}\label{Qpolicy}
\pi_{\bar{Q}}(y)=\arg\min_{a\in\{1,2\}}\bar{Q}(y,a), \qquad\forall\ y\in\mathcal{Y},
\end{equation}
where $\mathcal{Y}$ is the set of all possible $K$ latest observations and its size is $3^K$.
Clearly, \eqref{Qlearning} cannot be solved analytically and thus we obtain policy $\pi_{\bar{Q}}$ numerically using value iteration.

In Fig.\,\ref{Qapprox}, we plot the expected total discounted costs of policy $\pi_{\bar{Q}}$ for different values of $K\in\{1,2,3,4\}$ and compare these costs to the optimal policy $\pi_{opt}$ as functions
of the initial belief state $x$. Since Q-learning does not deal with belief states, we convert any initial $x$ into an
appropriate initial state $y(x)$ for Q-learning, by choosing the $y\in\mathcal{Y}$ to make the value of $x$
most probable:
$$y(x)=\min_{y\in\mathcal{Y}}\left\vert\Pr[H|y]-x\right\vert.$$
In Fig. \ref{Qapprox} we first observe the curves for small values of $K=1$ and 2. When initial belief state $x$ is close to 0 or 1, the gap between $\pi_{\bar{Q}}$  and $\pi_{opt}$ is more obvious.
This is because if $y$ has few elements with $y=1$ or $y=0$, it cannot approximate $x$ at the two extremes near 0 and 1. Since $\mathcal{Y}$ is a finite set, the corresponding values of $x$ are $Pr[H|y] \in [1-q, q]$ for all values of $y$, not containing $x = 0$ and $x=1$.
We further observe that as $K$ increases,
the expected total discounted cost of $\pi_{\bar{Q}}$  becomes closer to the optimal cost. 

We conclude that as $K$ increases, the Q-learning policy $\pi_{\bar{Q}}$  approximates the optimal policy more accurately.
To see this imagine that we run two versions of Q-learning, both using the belief state $x$
instead of the vector $y$, by discretising (finely) $[0,1]$ to make the state space $x\in\mathcal{X}$
finite (since Q-learning operates
over a finite set of states).
In the first version we use Bayesian updates for $x$.
In the second version we directly construct $x$ from the $K$-vector $y$ of last observations.
We expect as the state space of $x$ becomes finer, the first algorithm converges to the solution of the POMDP,
while the second algorithm converges to the original Q-learning based on $y$.
Now as $K\rightarrow\infty$, the value of $x$ used in both algorithms will tend to be the same.
Hence we expect as $K\rightarrow\infty$, Q-learning approaches the solution $\pi_{opt}$ of the POMDP.

\subsection{Information Restriction for Q-learning Platform}


We continue with the analysis of user incentives for the Q-learning Platform as in the case of the optimal platform.
Again we assume that users are sophisticated, have full information in how Q-learning works and can reverse-engineer
the Q-leaning policy $\pi_{\bar{Q}}$ to  decide whether to follow or not.
We define our IRM mechanism as before: the platform hides the history of user observations.
At each time, when the history of latest $K$ observations is $y\in\mathcal{Y}_1$ it recommends P1
and when $y\in\mathcal{Y}_2$ it recommends P2, as dictated by $\pi_{\bar{Q}}$. Here $\mathcal{Y}_a$ is the set of $y$ under which $\pi_{\bar{Q}}$ recommends action $a$.

As before, knowing $\pi_{\bar{Q}}$ and assuming that all users follow it,
a sophisticated user computes the asymptotic probability distribution
$P^{\pi_{\bar{Q}}}(y)$ of the last $K$ observation vector $y$.
Let $c_1(y)$ be the expected cost of taking action P1 given $y$,
\begin{equation}\label{c1y}
\small
c_1(y)=(\Pr[H|y]p+\Pr[L|y](1-p))c.
\end{equation}
If a user receives path recommendation P1, then she can infer that $y\in\mathcal{Y}_1$ and the expected cost of travelling through path P1 is $\sum\limits_{y\in\mathcal{Y}_1}P^{\pi_{\bar{Q}}}(y|y\in\mathcal{Y}_1)c_1(y)$.
This user will follow recommendation P1 if and only if
\begin{equation}\label{ic1}
\small
\sum\limits_{y\in\mathcal{Y}_1}P^{\pi_{\bar{Q}}}(y|y\in\mathcal{Y}_1)c_1(y)=
\sum\limits_{y\in\mathcal{Y}_1}\frac{P^{\pi_{\bar{Q}}}(y)c_1(y)}{\sum\limits_{y\in\mathcal{Y}_1}P^{\pi_{\bar{Q}}}(y)}\le c_M.
\end{equation}
Similarly, the user will follow  recommendation P2 if and only if
\begin{equation}\label{ic2}
\small
\sum\limits_{y\in\mathcal{Y}_2}P^{\pi_{\bar{Q}}}(y|y\in\mathcal{Y}_2)c_1(y)=
\sum\limits_{y\in\mathcal{Y}_2}\frac{P^{\pi_{\bar{Q}}}(y)c_2(y)}{\sum\limits_{y\in\mathcal{Y}_1}P^{\pi_{\bar{Q}}}(y)}> c_M.
\end{equation}
Therefore, given a fixed combination of system parameters' values, the Q-learning platform is incentive compatible if and only if both \eqref{ic1} and \eqref{ic2} hold.
An interesting question is to determine the range of parameters in the parameter space of the two-path model
for which incentive compatibility may not hold.


Fig. \ref{ICcm} examines the incentive compatibility (IC)  as a function of $c_M$ and $K$. We let $\beta=p=q=0.9$, $c=1$, and $c_M=0,0.001,\cdots,1$, and by solving \eqref{ic1} and \eqref{ic2} we find the regime of all possible $c_M$ values in which the IC does not hold.
 We observe that IC does not hold for all instances.
As $K$ increases,
the interval of values of $c_M$ in which IC does not hold becomes smaller. We also examine the incentive compatibility regarding $q$ in Fig. \ref{ICq} and regarding $\beta$ in Fig. \ref{ICbeta}.
In Fig. \ref{ICq}, we set $\beta=p=0.9$, $c_M=0.8$, $c=1$, and $q=0.5,0.501,\cdots,0.999$.
In Fig. \ref{ICbeta}), we set $p=q=0.9$, $c_M=0.8$, $c=1$, and $\beta=0.001,0.002,\cdots,0.999$.
As $K$ increases,
the interval of values of $\beta$ in which IC does not hold also becomes smaller. In all the three subfigures, we observe that the regime in which the IC does not hold becomes trivial once $K\geq 6$.

One may wonder the reason behind.
As $K$ increases, the Q-learning policy becomes more accurate as an approximation of the optimal policy and  by Theorem \ref{theorem}
IC holds for the optimal policy over all range of system parameters under IRM.
Thus, as the accuracy of the Q-learning policy increases, the information restriction mechanism should become `more'
incentive compatible in the sense that the instances for which IC does not hold become rare.

\section{Extension to a Multi-path Learning Model }\label{multipath}
In this section, we consider a more general network with three parallel paths where one more stochastic path P$1^\prime$ is added to our two-path model in Fig. \ref{tp}.
This new stochastic path P$1^\prime$ follows the same Markov model as path P1 in Fig. \ref{ts}.
Unlike our simple two-path model, we need to update the belief states of both stochastic paths now. Thus we use a belief state vector $x=(x^1,x^{1^\prime})$ whose updating follows the Bayesian inferencing process as in Section \ref{platform}.

We similarly denote the value function by $V(x^1,x^{1^\prime})$ with $V(x^1,x^{1^\prime})=V(x^{1^\prime},x^1)$ due to symmetry. Similar to \eqref{DP}, we define
$Q(x^1,x^{1^\prime},a)$ as the expected discounted cost staring from $x=(x^1,x^{1^\prime})$ if action $a$ is taken at the first time epoch and the optimal policy is followed thereafter.
$Q(x^1,x^{1^\prime},a)$ can be similarly written down as follows:
{\small
 \begin{align*}
&Q(x^1,x^{1^\prime},0)= x^1c_H+(1-x^1)c_L+\beta(x^1p_H+(1-x^1)p_L)\cdot\notag\\
&\quad V\bigg(\frac{x^1p_Hq_{HH}+(1-x^1)p_L(1-q_{LL})}{x^1p_H+(1-x^1)p_L},\\
&\quad x^{1^\prime}q_{HH}+(1-x^{1^\prime})(1-q_{LL})\bigg)+\nonumber\\
 &\quad \beta(x^1(1-p_H)+(1-x^1)(1-p_L))\cdot\notag\\
 &\quad V\bigg(\frac{x^1(1-p_H)q_{HH}+(1-x^1)(1-p_L)(1-q_{LL})}{x^1(1-p_H)+(1-x^1)(1-p_L)},\\
 &\quad x^{1^\prime}q_{HH}+(1-x^{1^\prime})(1-q_{LL})\bigg);\nonumber
 \end{align*}
  \begin{align*}
&Q(x^1,x^{1^\prime},1)=Q(x^{1^\prime},x_1,0);\nonumber\\
&Q(x^1,x^{1^\prime},1^\prime)=c_M+\beta V(x^1q_{HH}+(1-x^1)(1-q_{LL}),\\
&\quad x^{1^\prime}q_{HH}+(1-x^{1^\prime})(1-q_{LL})).
 \end{align*}}Similar to \eqref{optpolicy}, the optimality equation of
our three-path model is:
\begin{equation}\label{DP3links}
\small
V(x^1,x^{1^\prime})=\min_{a\in\{1,1^\prime,2\}}\{Q(x^1,x^{1^\prime},a)\}.
\end{equation}

Similar to Proposition \ref{property}, we can prove \eqref{DP3links} has a unique solution by
using the contraction mapping theorem.
As it is not in closed-form, we compute the value function using standard numerical methods such as value iteration.
We first discretise and partition the belief state space $[0,1]^2$ for $x^1, x^{1^\prime}$ equally into  $100\times100$ grids. In each iteration step, we directly evaluate the value function in each grid by solving \eqref{DP3links}.
Once the value function is obtained, the platform  computes the optimal policy for each
given belief state. We plot the optimal policy for problem \eqref{DP3links} in Fig.~\ref{optimal3link} where we let $p=q=\beta=0.9$, $c=1$ and $c_M=0.7$. We observe in Fig.~\ref{optimal3link} that here when cost belief $x^1$ ($x^{1^\prime}$) is small the optimal policy uses stochastic path P1 (P${1^\prime}$), and when both $x^1$ and $x^{1^\prime}$ are large the optimal policy uses deterministic path P2, which is similar to Proposition \ref{conj}.
We observe that the optimal policy is more complex and cannot be defined in terms of simple threshold rules.

 \begin{figure}
  \centering
  \includegraphics[width=2.1in]{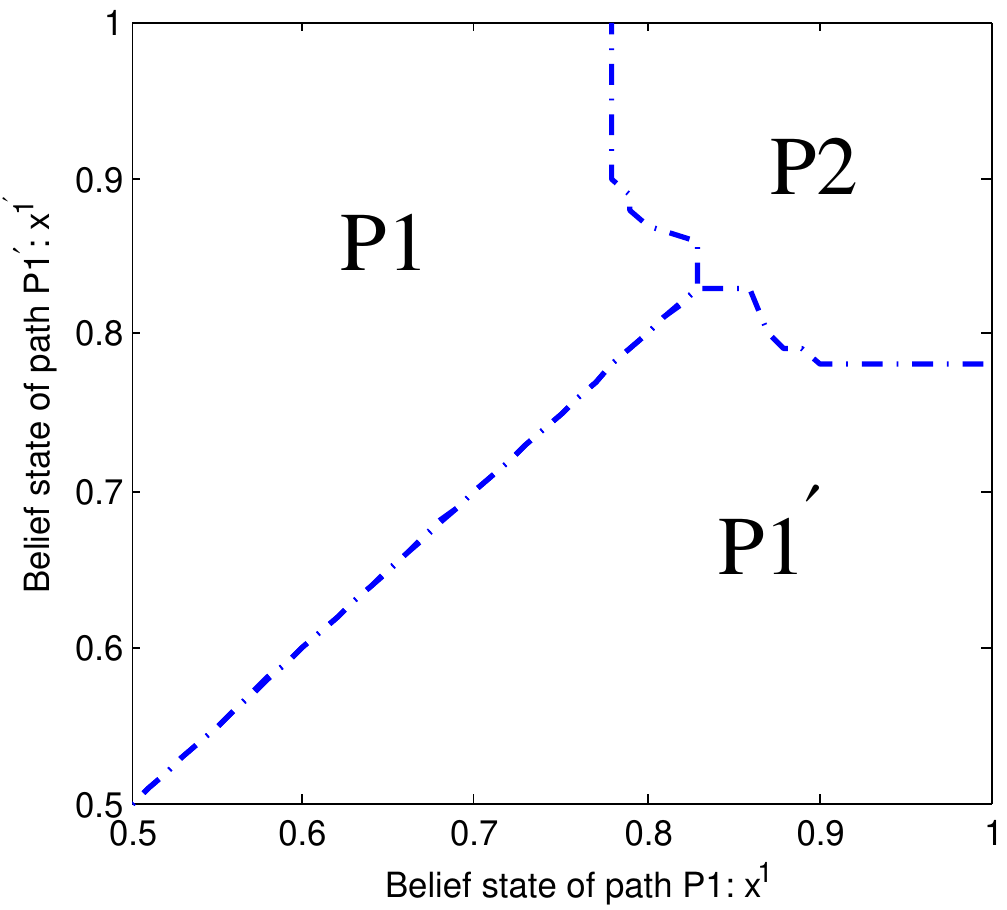}\vspace{-5pt}
  \caption{Optimal policy of path choices for  problem \eqref{DP3links}. We set $p=q=\beta=0.9$, $c=1$ and $c_M=0.7$.}
  \label{optimal3link}\vspace{-10pt}
\end{figure}

Without hiding any information, users will not follow the path recommendations of the optimal platform, and we can still use IRM for achieving incentive compatibility as in Theorem \ref{theorem}.
As for machine learning platforms, extensive numerical experiments show that Q-learning becomes incentive compatible (i.e., users following the platform suggestion is a correlated equilibrium) under IRM, for any value of $K$. Table \ref{tab:results} summarise the range of system parameters under which IC holds, by examining IC constraints in \eqref{ic1} and \eqref{ic2} for $K=1, 2, 3, 4$ and other parameter values exhaustively. We observe that the Q-learning platform is incentive compatible for all instances, which is different from Fig. \ref{IC} for two paths. With more stochastic paths, the quality of Q-learning algorithms improves and it has  close performance to the optimal platform for ensuring incentive compatibility.

\section{Conclusion}\label{conclusion}
\begin{table}[t]
\scriptsize
\centering
\caption{ The range of system parameters under which IC holds for the Q-learning platform. For the row of $c_M$ regime for IC, we let $p=q=\beta=0.9$, $c=1$, $c_M=0,0.01,\cdots,0.99,1$.
For the row of $q$ regime, we let $p=\beta=0.9$, $c=1$, $c_M=0.8$, $q=0.5,0.51,\cdots,0.99$.
For the row of $\beta$ regime,  we let $p=q=0.9$, $c=1$, $c_M=0.8$, $\beta=0.01,0.01=2,\cdots,0.99$. Note that $0\leq c_M\leq 1$, $0.5\leq q< 1$, and $0<\beta<1$. }
\begin{tabular}[tb]{|c|c|c|c|c|}\hline
$K$&1&2&3&4 \\\hline
$c_M$ regime&[0,1] &[0,1]&[0,1]&[0,1]\\\hline
$q$ regime&[0.5,0.99] &[0.5,0.99]&[0.5,0.99]&[0.5,0.99]\\\hline
 $\beta$ regime&[0.01,0.99] &[0.01,0.99]&[0.01,0.99]&[0.01,0.99]\\\hline
\end{tabular}
\label{tab:results}\vspace{-10pt}
\end{table}
In this paper we analyse incentive compatibility issues related to
users following recommendations by travel path optimizing platforms.
We show that socially optimal path recommendations based on past user travel cost history
are not always
incentive compatible since users like to myopically optimise their travel cost.
We discover the surprising result that if users have only access to the socially optimal platform
recommendations (besides full information on system parameters),
following these recommendation is a Nash equilibrium.
Numerical results suggest an interesting conjecture for practical platforms:
information hiding induces incentive compatibility for a wider range of system/network parameters as
the platform approximates closer the optimal platform (i.e., combining path exploration with path exploitation).

\appendices
\section{Proof of Proposition \ref{property}}\label{propertyproof}
First we prove that there is a unique value function $V(x)$ satisfying the optimality equation \eqref{DP} and it is continuous.

Let $V_k(x)$ be the value function of the $k$-stage problem, that is,
$$V_k(x)=\min_{\pi} E_\pi\left[\sum\limits_{t=1}^k \beta^{k-1}C(x_t,a_t)|x_1=x\right].$$
When $k=1$, we have
$$V_1=\min\{xc_H+(1-x)c_L,c_M\}.$$
For any $k\ge2$
  \begin{align}
V_k(x)=&\min\{xc_H+(1-x)c_L+\beta(xp+(1-x)(1-p))\cdot\nonumber\\
 &V_{k-1}\left(\frac{xpq+(1-x)(1-p)(1-q)}{xp+(1-x)(1-p)}\right)\nonumber\\
 &+\beta\big((1-x)p+x(1-p)\big)\cdot\nonumber\\
 &V_{k-1}\left(\frac{x(1-p)q+(1-x)p(1-q)}{x(1-p)+(1-x)p}\right),\nonumber\\
 &c_M+\beta V_{k-1}(xq+(1-x)(1-q))\}.\nonumber
 \end{align}

 Note that $V_1$ is a continuous function in $[0,1]$, that is, $V_1 \in \mathcal{C}[0,1]$. For any function $f\in \mathcal{C}[0,1]$, define
  \begin{align}
Tf(x)=&\min\{xc_H+(1-x)c_L+\beta(xp+(1-x)(1-p))\cdot\nonumber\\
 & f\left(\frac{xpq+(1-x)(1-p)(1-q)}{xp+(1-x)(1-p)}\right)\nonumber\\
 &+\beta\big((1-x)p+x(1-p)\big)\cdot\notag\\
& f\left(\frac{x(1-p)q+(1-x)p(1-q)}{x(1-p)+(1-x)p}\right),\nonumber\\
 &c_M+\beta f(xq+(1-x)(1-q))\}.\nonumber
 \end{align}
 Then $T$ is a map from $\mathcal{C}[0,1]$ to itself. It follows that
 $$V_k=T V_{k-1}=T^{k-1}V_1.$$
For any element $f$ in the space $\mathcal{C}[0,1]$, define the following norm,
$$||f||=\max\limits_{0\le x\le 1}|f(x)|.$$
With respect to this norm, $\mathcal{C}[0,1]$ is a Banach space. Note that $T$ is a contraction map, we can show that
$$||Tf-Tg||\le\beta||f-g||.$$
We need to first show that the following fact is true for any two functions $h$ and $l$ defined on the set $\{0,1\}$,
\begin{equation}\label{minmax}
  \min_{u\in\{0,1\}} h(u)-\min_{u\in\{0,1\}} l(u)\le\max_{u\in\{0,1\}} [h(u)-l(u)].
\end{equation}
Assume
$$u_h\in\arg\min_{u\in\{0,1\}} h(u),\qquad   u_l\in\arg\min_{u\in\{0,1\}} l(u).$$
Then,
\begin{align}
&\min_{u\in\{0,1\}} h(u)-\min_{u\in\{0,1\}} l(u)\notag\\
=&h(u_h)-l(u_l)\nonumber\\
=&h(u_h)-h(u_l)+h(u_l)-l(u_l)\nonumber\\
\le& h(u_l)-l(u_l)\le \max_{u\in\{0,1\}} [h(u)-l(u)].\nonumber
\end{align}
Apply \eqref{minmax}, we have for any $x\in[0,1]$
\begin{equation}
\begin{split}
&Tf(x)-Tg(x)\\
\le&\max\{\beta(xp+(1-x)(1-p))||f-g||\\
&+\beta(x(1-p)+(1-x)p)||f-g||, \beta||f-g||\}\\
=&\beta||f-g||.
\end{split}\nonumber
\end{equation}
Thus,
$$||Tf-Tg||\le\beta||f-g||.$$
By contracting mapping theorem, there is a unique element $V$ in $C[0,1]$ such that
$$V=\lim_{k\rightarrow \infty}V_k,$$
and that
$$V=TV.$$
Thus, there is a unique value function $V$ satisfying the DP equation. $V(x)$ is a continuous function of  $x$.

Next, we prove that $V(x)$ is an increasing function of  $x$. Note that $V_1(x)$ is an increasing function of $x$. Assume $V_{k-1}(x)$ is an increasing function of $x$, then we can prove that $V_{k}(x)$ is the minimum of two increasing functions of $x$. First, it is clear that
$$c_M+\beta V_{k-1}(xq+(1-x)(1-q))$$
is an increasing function of $x$ since $q\geq 1/2$. Also, $x c_H+(1-x)c_L$ is increasing in $x$. We only need to prove that
  \begin{align}
 &\beta(xp+(1-x)(1-p))\cdot\notag\\
 &V_{k-1}\left(\frac{xpq+(1-x)(1-p)(1-q)}{xp+(1-x)(1-p)}\right)+\nonumber\\
 &\beta\big((1-x)p+x(1-p)\big)\cdot\notag\\
 &V_{k-1}\left(\frac{x(1-p)q+(1-x)p(1-q)}{x(1-p)+(1-x)p}\right)\nonumber
 \end{align}
is increasing in $x$. Consider any $y\in[0,1]$ such that $y<x$, since $p,q\ge
\frac{1}{2}$, it is straightforward to prove that
$$\frac{xpq+(1-x)(1-p)(1-q)}{xp+(1-x)(1-p)}>\frac{ypq+(1-y)(1-p)(1-q)}{yp+(1-y)(1-p)},$$
\begin{align}
\frac{xpq+(1-x)(1-p)(1-q)}{xp+(1-x)(1-p)}>& \frac{x(1-p)q+(1-x)p(1-q)}{x(1-p)+(1-x)p}\notag\\
>&\frac{y(1-p)q+(1-y)p(1-q)}{y(1-p)+(1-y)p},\notag
\end{align}
and
$$xp+(1-x)(1-p)>yp+(1-y)(1-p).$$
Then, it follows from the induction hypothesis that
  \begin{align}
 &\beta(xp+(1-x)(1-p))V_{k-1}\left(\frac{xpq+(1-x)(1-p)(1-q)}{xp+(1-x)(1-p)}\right)\nonumber\\
 &+\beta\big((1-x)p+x(1-p)\big)\cdot\notag\\
 &V_{k-1}\left(\frac{x(1-p)q+(1-x)p(1-q)}{x(1-p)+(1-x)p}\right)\nonumber\\
 >&\beta(yp+(1-y)(1-p))V_{k-1}\left(\frac{ypq+(1-y)(1-p)(1-q)}{yp+(1-y)(1-p)}\right)\nonumber\\
 &+\beta((xp+(1-x)(1-p))-(yp+(1-y)(1-p)))\cdot\notag\\
 &V_{k-1}\left(\frac{y(1-p)q+(1-y)p(1-q)}{y(1-p)+(1-y)p}\right)+\nonumber\\
 &\beta((1-x)p+x(1-p))V_{k-1}\left(\frac{y(1-p)q+(1-y)p(1-q)}{y(1-p)+(1-y)p}\right)\nonumber\\
 =&\beta(yp+(1-y)(1-p))V_{k-1}\left(\frac{ypq+(1-y)(1-p)(1-q)}{yp+(1-y)(1-p)}\right)\nonumber\\
 &+\beta((1-y)p+y(1-p))\cdot\notag\\
 &V_{k-1}\left(\frac{y(1-p)q+(1-y)p(1-q)}{y(1-p)+(1-y)p}\right).\nonumber
 \end{align}
Thus, $V_k(x)$ is also an increasing function of $x$. By induction, for any $k$, $V_k(x)$ is  an increasing function of $x$. Given $x\ge y$, we have
\begin{align}
V(x)-V(y)=&\lim_{k\rightarrow\infty}V_k(x)-\lim_{k\rightarrow\infty}V_k(y)\notag\\
=&\lim_{k\rightarrow\infty}\left[V_k(x)-V_k(y)\right]\ge0.\notag
\end{align}

Finally, we prove the concavity of $V(x)$. One possible way is similar to the proof of Lemma 3.1 in Section III of \cite{ross1983introduction}:

Let $x=\lambda x_1+(1-\lambda)x_2$ where $0<\lambda<1$ and suppose
that the cost state of P1 is originally chosen as follows: A coin having probability
$\lambda$ of landing heads is flipped. If heads appears, then $H$ is chosen as the cost state with probability $x_1$ and if tails appears, then it is chosen
with probability $x_2$. Now the best that we can do if we are not to be told
the outcome of the coin flip is
$V(\lambda x_1+(1-\lambda)x_2)=V(x)$.
On the
other hand, if we are to be told the outcome of the flip, then our minimal
expected cost is $\lambda V(x_1)+(1-\lambda)V(x_2)$. Because this must be at
least as good as the case in which we are to be given no information
about the coin flip (one possible strategy is to ignore this information
apriori), we see that
$$\lambda V(x_1)+(1-\lambda)V(x_2)\le V(\lambda x_1+(1-\lambda)x_2)$$
which shows that $V(x)$ is concave.

Another way is to prove it by induction. Note that $V_1(x)$ is the minimum of two linear functions of $x$. Assume $V_{k-1}(x)$ is  the minimum of some collection of  linear functions of $x$ it follows that the same can be said of $V_k(x)$. Thus, by induction, $V_k(x)$ is a concave function of $x$ and
$$\lambda V_k(x_1)+(1-\lambda)V_k(x_2)\le V_k(\lambda x_1+(1-\lambda)x_2)$$
holds for every $k$. Thus, $V(x)$ is concave.

\section{Proof of Proposition \ref{conj}}\label{conjproof}
Assume $ \beta(2q-1)<2/3$. To prove Proposition \ref{conj}, it suffices to prove the following statement: for any $x$, $y$ $\in[0,1]$ such that $x>y$,
 \begin{equation}\label{sufficient}
Q(x,1)-Q(y,1)>Q(x,2)-Q(y,2).
\end{equation}
This is because if $\eqref{sufficient}$ holds for any $x>y$, the difference of the travel costs of P1 and P2 are monotone, i.e.,
\begin{equation}\label{differ}
Q(x,1)-Q(x,2)
\end{equation}
is strictly  increasing in $x$. Note that when $x=0$, \eqref{differ} is less or equal to 0 and when  $x=1$, \eqref{differ} is nonnegative. Since $V(x)$ is continuous, there is a unique threshold value $x^*$ such that \eqref{differ} is equal to zero and it is optimal to use P1 when $x\le x^*$ and use P2 when $x\ge x^*$ .
Now we consider three cases.\\
Case 1:  $ x\ge 1/2$ and $y\le1/2$.
 Define
$$A(x,y)=Q(x,1)-Q(y,1),\qquad B(x,y)=Q(x,2)-Q(y,2).$$
By optimality equation \eqref{DP} we can prove that
\begin{align}\label{difference}
&V(x)-V(y)=TV(x)-TV(y)\notag\\
&\in[\min\{A(x,y),B(x,y)\},\max\{A(x,y),B(x,y)\}].
\end{align}
We prove \eqref{difference} by considering the following four cases:

i) If $V(x)=Q(x,1)$ and $V(y)=Q(y,1)$, it is clear that \eqref{difference} holds.

ii) If  $V(x)=Q(x,2)$ and $V(y)=Q(y,2)$, it is clear that \eqref{difference} holds.

iii) If $V(x)=Q(x,1)\le Q(x,2)$ and $V(y)=Q(y,2)\le Q(y,1)$, it follows that
$$Q(x,1)-Q(y,1)\le V(x)-V(y)\le Q(x,2)-Q(y,2)$$
and \eqref{difference} holds.

iv) If $V(x)=Q(x,2)\le Q(x,1)$ and $V(y)=Q(y,1)\le Q(y,2)$, it follows that
$$Q(x,2)-Q(y,2)\le V(x)-V(y)\le Q(x,1)-Q(y,1)$$
and \eqref{difference} holds.

Since $ x\ge 1/2 \ge y$, we have
$$x\ge xq+(1-x)(1-q)\ge1/2\ge yq+(1-y)(1-q)\ge y.$$
We have already proven that $V(x)$ is an increasing function of $x$, thus
$$V(x)\ge V(xq+(1-x)(1-q))\ge V(yq+(1-y)(1-q))\ge V(y).$$
It follows that
\begin{align}
&V(x)-V(y)\notag\\
\ge& V(xq+(1-x)(1-q))-V(yq+(1-y)(1-q))\notag\\
\ge&\beta V(xq+(1-x)(1-q))-\beta V(yq+(1-y)(1-q)).\notag
\end{align}
If the equality holds, then $V(z)$ is constant for any $z\ge y$. Then the proposition follows.
Now we assume the equality does not hold, i.e.,
\begin{align}
B(x,y)=&\beta V(xq+(1-x)(1-q))-\beta V(yq+(1-y)(1-q))\notag\\
<&V(x)-V(y)\le\max\{A(x,y),B(x,y)\}.\notag
\end{align}
As a consequence, $A(x,y)>B(x,y)$.\\
Case 2: $x\le 1/2$ and $y\le1/2$. We have
\begin{equation}
\begin{cases}
y\le x\le xq+(1-x)(1-q),\\
y\le yq+(1-y)(1-q)\le xq+(1-x)(1-q).
\end{cases}\notag
\end{equation}
If $x\ge yq+(1-y)(1-q)$, from the concavity of $V(x)$ it follows that
$$V(yq+(1-y)(1-q))-V(y)\ge V(xq+(1-x)(1-q))-V(x).$$
If $x\le yq+(1-y)(1-q)$,  from the concavity of $V(x)$ it follows that
$$V(x)-V(y)\ge V(xq+(1-x)(1-q))-V(yq+(1-y)(1-q)).$$
Similar to Case 1, we can prove that $\eqref{sufficient}$ holds.\\
Case 3: $y\ge 1/2$. First, we prove
\begin{equation}\label{derivative}
V(x)-V(y)\le\frac{c_H-c_L }{1-\beta(2q-1)}(x-y)
\end{equation}
holds for any $x>y\ge 1/2$ by induction. It is clear that \eqref{derivative} holds for $V_1(x)$. Assume \eqref{derivative} holds for $V_{k-1}(x)$. We can define $A_k(x,y)$ and $B_k(x,y)$ similarly as in Case 1 and similar to \eqref{difference}, we can prove that
\begin{align}\label{differencek}
&V_k(x)-V_k(y)=TV_{k-1}(x)-TV_{k-1}(y)\in\notag\\
&[\min\{A_{k-1}(x,y),B_{k-1}(x,y)\},\max\{A_{k-1}(x,y),B_{k-1}(x,y)\}].\nonumber
\end{align}\\
From the induction hypothesis, it follows that
\begin{align}
&A_{k-1}(x,y)\notag\\
=&xc_H+(1-x)c_L+\beta(xp+(1-x)(1-p))\cdot\notag\\
&V_{k-1}\left(\frac{xpq+(1-x)(1-p)(1-q)}{xp+(1-x)(1-p)}\right)+\notag\\
&\beta\big((1-x)p+x(1-p)\big)V_{k-1}\left(\frac{x(1-p)q+(1-x)p(1-q)}{x(1-p)+(1-x)p}\right)\notag\\
&-(yc_H+(1-y)c_L)-\beta(yp+(1-y)(1-p))\cdot\notag\\
&V_{k-1}\left(\frac{ypq+(1-y)(1-p)(1-q)}{yp+(1-y)(1-p)}\right)-\notag\\
&\beta\big((1-y)p+y(1-p)\big)V_{k-1}\left(\frac{y(1-p)q+(1-y)p(1-q)}{y(1-p)+(1-y)p}\right)\notag\\
\le&(x-y)(c_H-c_L)+\beta(2q-1)\frac{c_H-c_L }{1-\beta(2q-1)}(x-y)\notag\\
\le&\frac{c_H-c_L }{1-\beta(2q-1)}(x-y),\notag
\end{align}
and that
\begin{equation}
\begin{split}
B_{k-1}(x,y)=&c_M+\beta V_{k-1}(xq+(1-x)(1-q))-\\
&\left(c_M+\beta V_{k-1}(yq+(1-y)(1-q))\right)\\
\le&\beta(2q-1)(x-y)\frac{c_H-c_L }{1-\beta(2q-1)}\\
\le&\frac{c_H-c_L }{1-\beta(2q-1)}(x-y).
\end{split}\nonumber
\end{equation}
Thus, \eqref{derivative} holds for $V_k(x)$. Therefore, \eqref{derivative} holds for any $x>y\ge 1/2$.

From the concavity of $V(x)$, it follows that
\begin{equation}
\begin{split}
&\left(Q(x,1)-Q(x,2)\right)-\left(Q(y,1)-Q(y,2)\right)\\
=&(x-y)(c_H-c_L)+\beta(yp+(1-y)(1-p))\cdot\\
&\Biggl(V\left(\frac{xpq+(1-x)(1-p)(1-q)}{xp+(1-x)(1-p)}\right)-\\
&V\left(\frac{ypq+(1-y)(1-p)(1-q)}{yp+(1-y)(1-p)}\right)\Biggl)+\\
&\beta((2p-1)(x-y))\Biggl(V\left(\frac{xpq+(1-x)(1-p)(1-q)}{xp+(1-x)(1-p)}\right)-\\
&V\left(\frac{y(1-p)q+(1-y)p(1-q)}{y(1-p)+(1-y)p}\right)\Biggl)+\\
&\beta(x(1-p)+(1-x)p)\Biggl(V\left(\frac{x(1-p)q+(1-x)p(1-q)}{x(1-p)+(1-x)p}\right)\\
&-V\left(\frac{y(1-p)q+(1-y)p(1-q)}{y(1-p)+(1-y)p}\right)\Biggl)-\\
&\beta(V(xq+(1-x)(1-q))-V(yq+(1-y)(1-q)))\\
\ge& (x-y)(c_H-c_L)-\frac{1}{2}\beta\big(V(xq+(1-x)(1-q))-\\
&V(yq+(1-y)(1-q))\big).
\end{split}\nonumber
\end{equation}
From \eqref{derivative} and $\beta(2q-1)<2/3$,  it follows that
\begin{equation}
\begin{split}
&\left(Q(x,1)-Q(x,2)\right)-\left(Q(y,1)-Q(y,2)\right)\\
\ge& (x-y)(c_H-c_L)-\frac{1}{2}\beta\big(V(xq+(1-x)(1-q))-\\
&V(yq+(1-y)(1-q))\big)\\
\ge&(x-y)(c_H-c_L)-\frac{1}{2}\beta(2q-1)(x-y)\frac{c_H-c_L }{1-\beta(2q-1)}\\
=&(x-y)(c_H-c_L)\frac{1-\frac{3}{2}\beta(2q-1)}{1-\beta(2q-1)}\\
>&0.
\end{split}\nonumber
\end{equation}
Thus, \eqref{sufficient} holds for any $x>y\ge 1/2$ and the proposition follows.

\section{Proof of Proposition \ref{poapim}}\label{poapimproof}
   When $\beta=0$, the optimal policy is the same as the myopic policy. Thus, $PoA=\frac{1}{1-\beta}=1$ and the proposition holds. We will assume $\beta\in(0,1)$ and  first show that the price of anarchy must  be larger than or equal to $\frac{1}{1-\beta}$.

   Let $p=1$ and  $q\in(1/2,1)$. Then $c_L=0$, $c_H=c$ and the optimality equation \eqref{DP} can be written as
\begin{equation}\label{DPPoA2}
\begin{split}
V(x)=&\min\{x+\beta (xV(q)+(1-x)V(1-q)),\\
&c_M+\beta V(xq+(1-x)(1-q))\}.
\end{split}
\end{equation}

Since $c_M>0$, we can choose $q$ close enough to 1 and $c$ large enough such that  $(1-q)c<c_M<c/2$. We will compute value function of the myopic policy, i.e., $V_{\pi_m}(x)$. When $xc\le c_M$, myopic policy chooses path P1. Then,
$$V_{\pi_m}(x)=xc+\beta (xV_{\pi_m}(q)+(1-x)V_{\pi_m}(1-q)).$$
When $xc> c_M$, myopic policy chooses path P2. Then,
$$V_{\pi_m}(x)=c_M+\beta V_{\pi_m}(xq+(1-x)(1-q))=\frac{c_M}{1-\beta}.$$

Next we will bound the value function of the optimal policy. Note that
$$V(x)\le xc+\beta (xV(q)+(1-x)V(1-q)),$$
and
$$V(x)\le \frac{c_M}{1-\beta}.$$
It follows that,
\begin{equation}
\left\{
\begin{array}{ll}
V(q)&\le \frac{c_M}{1-\beta},\\
V(1-q)&\le(1-q)c+\beta(1-q)V(q)+\beta qV(1-q).
\end{array}
\right.\nonumber
\end{equation}
Thus,
 \begin{equation}
\left\{
\begin{array}{ll}
V(q)&\le \frac{c_M}{1-\beta},\\
V(1-q)&\le \frac{(1-q)c+\beta(1-q)\frac{c_M}{1-\beta}}{1-\beta q},
\end{array}
\right.\nonumber
\end{equation}
and
$$V(x)\le (\frac{(1-\beta)c+\beta c_M}{1-\beta})(x+(1-x)\frac{\beta(1-q)}{1-\beta q}).$$

Choose a small enough positive number $\epsilon$, let
$$\frac{(1-q)(c_M+\epsilon)}{c_M}<x<\frac{1}{2}$$
and $c=\frac{c_M+\epsilon}{x}$.
Note that such $x$ and $c$ satisfy $(1-q)c<c_M<c/2$ and $xc>c_M$. It follows that
$$V_{\pi_m}(x)=\frac{c_M}{1-\beta},$$
and
$$V(x)\le(\frac{(1-\beta)c+\beta c_M}{1-\beta})(x+(1-x)\frac{\beta(1-q)}{1-\beta q}).$$
Thus,
\begin{align}
PoA\ge& \frac{\frac{c_M}{1-\beta}}{(\frac{(1-\beta)c+\beta c_M}{1-\beta})(x+(1-x)\frac{\beta(1-q)}{1-\beta q})}\notag\\
\xrightarrow{q\rightarrow1}&\frac{\frac{c_M}{1-\beta}}{\frac{(1-\beta)(c_M+\epsilon)+\beta c_Mx}{1-\beta}}\xrightarrow{x,\epsilon\rightarrow0}\frac{1}{1-\beta}.\notag
\end{align}
Next we will show that the price of anarchy must  be less than or equal to $\frac{1}{1-\beta}$.

Note that for any $x$,
\begin{equation}\label{ineq2}
V_{\pi_m}(x)\le\frac{c_M}{1-\beta}.
\end{equation}
We can prove \eqref{ineq2} by induction or by arguing that the cost for each time period can not be larger than $c_M$ since the myopic policy always chooses the path with minimal cost.

Let $V_{\pi_m,k}(x)$ be the $k$-stage cost of myopic policy,
then,
$$V_{\pi_m,1}=\min\{xc_H+(1-x)c_L,c_M\},$$
and for any $k\ge2$, if $xc_H+(1-x)c_L\le c_M$, then
  \begin{align}
V_{\pi_m,k}(x)=&xc_H+(1-x)c_L+\beta(xp+(1-x)(1-p))\cdot\nonumber\\
 &V_{\pi_m,k-1}\left(\frac{xpq+(1-x)(1-p)(1-q)}{xp+(1-x)(1-p)}\right)+\nonumber\\
 &\beta\big((1-x)p+x(1-p)\big)\cdot\notag\\
 &V_{\pi_m,k-1}\left(\frac{x(1-p)q+(1-x)p(1-q)}{x(1-p)+(1-x)p}\right),\nonumber
 \end{align}
if $xc_H+(1-x)c_L>c_M$, then
 $$V_{\pi_m,k}(x)=c_M+\beta V_{\pi_m,k-1}(xq+(1-x)(1-q)).$$
Similar to the proof of Proposition \ref{property}, we can prove that
$$V_{\pi_m}(x)=\lim_{k\rightarrow\infty}V_{\pi_m,k}(x)$$
by contracting mapping theorem.

Then we prove by induction that for any $k$ and $x$,
\begin{equation}\label{ineq1}
V_{\pi_m,k}(x)\le \frac{V_k(x)}{1-\beta}.
\end{equation}
When $k=1$,  \eqref{ineq1} follows from the fact
$$V_{\pi_m,1}(x)=V_1(x).$$
Assume that \eqref{ineq1} holds for $k-1$ and  any $x$. We need to show that it also holds for $k$ and any $x$. Given any $x\in[0,1]$, if $xc_H+(1-x)c_L>c_M$, then it follows from \eqref{ineq2} that
$$V_k(x)\ge\min\{xc_H+(1-x)c_L,c_M\}=c_M\ge (1-\beta)V_{\pi_m,k}(x).$$
Then \eqref{ineq1} holds. If $xc_H+(1-x)c_L\le c_M$, then
  \begin{align}
V_{\pi_m,k}(x)=&xc_H+(1-x)c_L+\beta(xp+(1-x)(1-p))\cdot\nonumber\\
 &V_{\pi_m,k-1}\left(\frac{xpq+(1-x)(1-p)(1-q)}{xp+(1-x)(1-p)}\right)+\nonumber\\
 &\beta\big((1-x)p+x(1-p)\big)\cdot\notag\\
 &V_{\pi_m,k-1}\left(\frac{x(1-p)q+(1-x)p(1-q)}{x(1-p)+(1-x)p}\right),\nonumber
 \end{align}
 and
   \begin{align}
V_{k}(x)=&xc_H+(1-x)c_L+\beta(xp+(1-x)(1-p))\cdot\nonumber\\
 &V_{k-1}\left(\frac{xpq+(1-x)(1-p)(1-q)}{xp+(1-x)(1-p)}\right)+\nonumber\\
 &\beta\big((1-x)p+x(1-p)\big)\cdot\notag\\
 &V_{k-1}\left(\frac{x(1-p)q+(1-x)p(1-q)}{x(1-p)+(1-x)p}\right),\nonumber
 \end{align}
It follows from the induction hypothesis that
$$\frac{V_{\pi_m,k}(x)}{V_{k}(x)}\le\frac{1}{1-\beta}.$$
Thus, \eqref{ineq1} holds for any $k$ and  any $x$. Then by letting $k\rightarrow\infty$, we get
$$\frac{V_{\pi_m}(x)}{V(x)}\le\frac{1}{1-\beta}.$$
Therefore, $PoA\le \frac{1}{1-\beta}$.

 \section{Proof of Proposition \ref{poapie}}\label{poapieproof}
We consider the similar instance in the proof of Proposition \ref{poapim} where $p=1$ and $q\in(1/2,1)$. Then $c_L=0$, $c_H=c$. Here, different from  the proof of Proposition \ref{poapim}, we let $x=0$, and $c>2c_M$, then
$$V_{\pi_\emptyset}(x)=\frac{c_M}{(1-\beta)}.$$
If $\beta\in(0,1)$, according to the proof of Proposition \ref{poapim}
\begin{align}
V(x)\le&(\frac{(1-\beta)c+\beta c_M}{1-\beta})(x+(1-x)\frac{\beta(1-q)}{1-\beta q})\notag\\
=&\frac{(1-\beta)c+\beta c_M }{1-\beta}\frac{\beta(1-q)}{1-\beta q}.\notag
\end{align}
It follows that
$$PoA\ge \frac{\frac{c_M}{(1-\beta)}}{\frac{(1-\beta)c+\beta c_M }{1-\beta}\frac{\beta(1-q)}{1-\beta q}}\xrightarrow{q\rightarrow1}\infty.$$
If $\beta=0$, then $V(x)=0$. Thus, $PoA=\infty$.

\end{document}